\documentclass[11pt,letterpaper]{article}
\usepackage[margin=1in]{geometry}
\usepackage{fullpage}
\usepackage{float}
\usepackage{amsfonts,amsmath,amssymb,amsthm,graphicx}

\usepackage{amstext}
\usepackage{url}
\usepackage{subfig}
\usepackage{hyperref}
\usepackage{multirow}
\usepackage{caption}
\usepackage{comment}
\usepackage[capitalize]{cleveref}
\usepackage[authoryear,round]{natbib}
\usepackage{enumerate}

\allowdisplaybreaks[1]

\usepackage[linesnumbered,ruled]{algorithm2e}
\SetKwComment{Comment}{/* }{ */}

\newtheorem{theorem}{Theorem}[section]
\newtheorem{corollary}[theorem]{Corollary}
\newtheorem{lemma}[theorem]{Lemma}
\newtheorem{proposition}[theorem]{Proposition}

\newtheorem{definition}[theorem]{Definition}
\newtheorem{assumption}[theorem]{Assumption}

\newtheorem{observation}[theorem]{Observation}

\newtheorem{example}[theorem]{Example}

\newcommand{\floor}[1]{
{\lfloor {#1} \rfloor}
}

\newcommand{\given}{\,|\,}

\newcommand{\prob}[2][]{\text{\bf Pr}\ifthenelse{\not\equal{}{#1}}{_{#1}}{}\!\left[{\def\givenn{\middle|}#2}\right]}
\newcommand{\expect}[2][]{\text{\bf E}\ifthenelse{\not\equal{}{#1}}{_{#1}}{}\!\left[{\def\givenn{\middle|}#2}\right]}

\newcommand{\tparen}{\big}
\newcommand{\tprob}[2][]{\text{\bf Pr}\ifthenelse{\not\equal{}{#1}}{_{#1}}{}\tparen[{\def\given{\tparen|}#2}\tparen]}
\newcommand{\texpect}[2][]{\text{\bf E}\ifthenelse{\not\equal{}{#1}}{_{#1}}{}\tparen[{\def\given{\tparen|}#2}\tparen]}

\newcommand{\sprob}[2][]{\text{\bf Pr}\ifthenelse{\not\equal{}{#1}}{_{#1}}{}[#2]}
\newcommand{\sexpect}[2][]{\text{\bf E}\ifthenelse{\not\equal{}{#1}}{_{#1}}{}[#2]}

\newcommand{\eps}{\epsilon}

\newcommand{\rev}{{\rm Rev}}
\newcommand{\poly}{{\rm poly}}
\newcommand{\opt}{{\rm OPT}}
\newcommand{\optval}{\rm{OPT\text{-}EF}}
\newcommand{\optefw}{\rm{OPT\text{-}EF1}}
\newcommand{\opteps}{{\rm OPT\text{-}\epsilon EF}}

\usepackage{xcolor}

\usepackage{array}
\usepackage{booktabs}

\usepackage{color-edits}
\addauthor{yl}{red}

\usepackage{thm-restate}

\DeclareMathOperator*{\argmin}{arg\,min}

\renewcommand{\P}{\mbox{\sf P}}
\newcommand{\NP}{\mbox{\sf NP}}

\begin{document}

\title{Algorithmic Fair Contracts}

\author{
Matteo Castiglioni\thanks{DEIB, Politecnico di Milano. Email: \texttt{matteo.castiglioni@polimi.it}}
\and Junjie Chen\thanks{Department of Economics, National University of Singapore. Email: \texttt{junjchen9-c@my.cityu.edu.hk}}
\and Yingkai Li\thanks{Department of Economics, National University of Singapore. Email: \texttt{yk.li@nus.edu.sg} }
}


\date{}

\begin{titlepage}
	\clearpage\maketitle
	\thispagestyle{empty}

\begin{abstract}
We initiate the algorithmic study of fair contract design. A principal assigns multiple tasks to heterogeneous agents and chooses task-level linear contracts; agents differ in costs and success probabilities, and fairness requires each agent to prefer her own task-contract bundle to any other agent's. Unlike envy-free allocations of indivisible items, envy-free full-allocation contracts always exist, but optimizing revenue under this constraint is computationally difficult: no polynomial-time algorithm can achieve any constant-factor approximation in general. We therefore identify tractable regimes. With a constant number of tasks, optimal EF, EF1, and $\epsilon$-EF contracts are computable in polynomial time. With a constant number of agents, exact EF remains hard, even for three agents, while EF1 and $\epsilon$-EF admit additive FPTAS against the EF benchmark. We also show that exact EF can have an unbounded price of fairness, whereas $\epsilon$-EF and EF1 can restore bounded revenue loss.
\end{abstract}

\end{titlepage}

\section{Introduction}

Contract theory studies how a principal can use payments to induce agents to take costly actions that are not directly observable \citep{grossman1992analysis,bolton2004contract}. This incentive problem has recently become a central topic in theoretical computer science, where algorithmic work has examined the complexity of optimal contracts \citep[e.g.,][]{guruganesh2021contracts, castiglioni2025reduction,dutting2023multi}, approximation guarantees for simple contract classes \citep[e.g.,][]{dutting2019simple,alon2022bayesian,ezra2023approximability}, and learning contracts from data \citep[e.g.,][]{bacchiocchi2023learning,chen2024bounded}; see the survey by \citet{dutting2024algorithmic}.

Most of this literature optimizes the principal's revenue subject to individual incentive requirements: the contract must make the prescribed effort optimal for the agent and give the agent enough utility to take the task. These requirements do not address a different agent-side question: when different agents face different contracts, are these contracts fair relative to one another? We study this comparative question. In particular, we ask whether each agent would prefer her own contract to the contract offered to any other agent, evaluated using her own costs, success probabilities, and induced utility.

This perspective is increasingly relevant in digital labor markets and algorithmic management, where platforms such as Uber, Didi, DoorDash, Deliveroo, and Amazon Mechanical Turk use automated systems to set terms of work, prices, incentives, schedules, and sometimes assignments. Such decisions are scrutinized through the lens of fair working conditions, nondiscrimination, and transparency in compensation and work organization~\citep[e.g.,][]{employment2019draft}; the EU Platform Work Directive, for example, explicitly identifies automated systems that allocate tasks, price individual assignments, determine schedules, and provide incentives as central features of platform work~\citep{eu2024platform}. These concerns also connect to broader questions about inequality in labor markets~\citep[e.g.,][]{chipriyanov2024researching,wbg24path}. Recent empirical work on ride-hailing studies algorithmic pay and work allocation together, documenting that dynamic pricing can affect both job predictability and inequality across drivers~\citep{binns2025uber}. These examples motivate the study of contract design under fairness constraints, rather than revenue maximization alone.

Envy-freeness provides a natural formalization of this comparative concern. It is  a utility-based comparison different from equality of nominal payments: an agent compares the utility she obtains from her own contract with the utility she would obtain from someone else's contract. This comparison is about expected utility, not just nominal pay. A long trip, a risky delivery, a specialized data-labeling job, or a task requiring costly preparation may be much less attractive than another contract with the same posted payment. Conversely, two agents can value the same contractual terms differently because their success probabilities and effort costs differ. Thus, simply imposing equal-pay contracts is not enough. Equal pay is an exogenous restriction on transfers; it may leave agents with unequal net utilities, may fail to remove envy, and can even be incompatible with envy-free full allocation in simple instances (see \cref{apx:equal_pay_not_fair}). 

This paper introduces \emph{fair contracts}: contract-design problems in which the principal optimizes incentives and revenue while respecting agents' fairness comparisons. We focus on envy-freeness (EF): each agent should weakly prefer her own contract to that assigned to any other agent, evaluated using her own subjective utility. Our formal model instantiates this idea in a multi-task environment, where the principal chooses both a task allocation and task-level linear incentive payments. Our goal is to understand when fair contracts exist, how hard it is to optimize revenue subject to fairness, and how much revenue the principal may lose by imposing fairness constraints.

\subsection{Model}
We now instantiate fair contracts in a simple multi-task moral-hazard environment. A principal faces $m$ tasks and $n$ heterogeneous agents. If task $j$ is performed by agent $i$ and the agent exerts effort, the task succeeds with probability $p_{i,j}$, generates reward $r_j$ for the principal, and costs the agent $c_{i,j}$. In this setting, a contract consists of task-level linear shares $\alpha_j\in[0,1]$ and a delegation of each task to one agent, where $\alpha_j$ specifies the fraction of the reward paid to the delegated agent upon success.

A contract is fair if no agent envies another agent's contract. In this model, each agent's contract is her bundle of delegated tasks together with the linear shares attached to those tasks. Because agents have different success probabilities and costs, the same contract can have different values to different agents. Thus, agent $i$ compares her own contract with agent $j$'s by evaluating both using agent $i$'s own success probabilities and effort costs. The principal's objective is to maximize expected revenue subject to these envy-free constraints.

We focus on full-allocation contracts: every task is assigned, and each assigned contract is large enough to make effort worthwhile for the assigned agent. This rules out the vacuous fair solution that sets all contracts to zero and induces no effort. This restriction preserves the natural analogy with fair division, where all items must be allocated \citep[e.g.,][]{amanatidis2023fair}. We can show that under a mild usefulness assumption on tasks, envy-free full-allocation contracts exist.

Existence, however, does not make optimization easy. Revenue maximization can conflict sharply with envy-freeness. Even with a single task, the highest-revenue contract may also be the most attractive one from another agent's perspective; satisfying envy-freeness may therefore require changing the contract and sacrificing revenue. The main question of the paper is how severe this conflict is, computationally and quantitatively.

Fair contracts are related to, but distinct from, fair division and envy-free pricing. In fair division, an allocated good or chore directly determines utility, and the main issue is the fairness of the allocation itself. Here, utility is mediated by incentives: an opportunity has value only through a contract that makes costly effort attractive, as in moral-hazard contract theory~\citep{grossman1992analysis,bolton2004contract}. Thus, allocation, payment, and hidden action interact. This difference also changes the existence question: envy-free allocations of indivisible items may fail to exist~\citep[e.g.,][]{aziz2022fair,amanatidis2023fair}, whereas envy-free full-allocation contracts always exist in our model. The challenge is therefore not finding any fair contract, but finding a revenue-optimal one.

The model is also different from envy-free pricing~\citep[e.g.,][]{guruswami2005profit,chen2010envy,anshelevich2017envy}. In envy-free pricing, a seller posts prices and buyers select their favorite bundles; prices enter buyers' utilities symmetrically. In fair contracts, the same incentive share affects agents differently because success probabilities and effort costs are heterogeneous, and the principal must choose both the contracts and the allocated bundles of tasks for agents. This additional incentive dimension is what drives the computational hardness in our setting.

\subsection{Our Results}
\begin{table}[t]
\renewcommand{\arraystretch}{1.5}
\centering
\caption{Summary of computational and price of fairness results. {\bf The first two result columns use optimal EF contracts as benchmarks}. Approximation ratios are benchmark revenue divided by achieved revenue, and hence are at least one.}
\label{tb:hardness}
\begin{tabular}{|m{1.6cm}<{\centering} | m{3.4cm}<{\centering}  |m{3.5cm}<{\centering} | m{3.6cm}<{\centering} | m{2.2cm}<{\centering} | }
\hline
 & \textbf{General settings} & \textbf{Constant number of agents} & \textbf{Constant number of tasks} 
& \textbf{Price of fairness}\\
\hline
\multirow{2}{*}{\parbox{1.7cm}{ \centering EF contracts} } & \multirow{2}{*}{\parbox{3.5cm}{ \centering No constant approx} } & No approx $<\frac{5}{2}$ for $3$ agents & \multirow{2}{*}{\parbox{3.8cm}{\centering EF contracts are polynomially solvable}} 
& \multirow{2}{*}{ \parbox{2cm}{\centering unbounded\,\,} }
\\
\cline{3-3}
                         &                          & \NP-hard for $2$ agents& & \\
\hline
$\eps$-EF contracts & 
\parbox{3.6cm}{\centering No constant approx, even for a constant $\epsilon$} & Additive FPTAS & 
\parbox{3.8cm}{\centering $\eps$-EF contracts are polynomially solvable}  
& $[\frac{1}{4\eps},\frac{\min\{m,n^2\}}{4\eps}]$\\
\hline
{EF1 contracts } & { \centering Unknown} & Additive FPTAS &  
\parbox{3.8cm}{\centering EF1 contracts are polynomially solvable }
& $[\sqrt{n},n^2]$
\\
\hline
\end{tabular}
\end{table}

A key conceptual contribution of this paper is to bring agent-side fairness into contract design and initiate the study of algorithmic fair contracts. Fairness constraints change the principal's problem in a fundamental way: the contract must not only induce effort and generate revenue, but also withstand comparisons across agents. In the multi-task model we study, this creates a new interaction among incentives, fairness, and profit maximization.

We now summarize our main results. They show that this new fairness constraint is structurally benign but algorithmically demanding: envy-free full-allocation contracts always exist, yet optimizing revenue under envy-freeness is hard in general. Table~\ref{tb:hardness} summarizes the main computational results. 

\vspace{1.5mm}

\noindent{\bf Strong negative results in general settings.}
For arbitrary numbers of agents and tasks, the optimal EF benchmark admits no constant approximation ratio in polynomial time. The hardness is robust: for every constant ratio $c\ge 1$, there is a constant $\epsilon>0$ such that even an $\epsilon$-EF contract cannot guarantee a $1/c$ fraction of the EF benchmark (\cref{def:eps_ef,arbitryri_hardnesdd_linear}). The reduction is from a gap version of bounded-degree independent set \citep{alon1995derandomized,trevisan2001non}. High revenue requires assigning many vertex tasks to one efficient agent, while envy-freeness forces those tasks to form an independent set.

\vspace{1.5mm}

\noindent{\bf Algorithms and hardness in restricted settings.}
The general hardness motivates two natural restrictions. When the number of tasks is constant, the assignment space is polynomially enumerable. For EF and $\epsilon$-EF, each fixed assignment yields a linear program. For EF1, the difficulty is that one must also identify, for each potential envy comparison, the task whose removal eliminates envy. We show how to avoid an exponential enumeration by using  break-even shares as contract upper bounds, obtaining polynomial-time algorithms for optimal EF1 contracts (\cref{thm:constant_task}).

When the number of agents is constant, exact EF remains computationally difficult. Even with three agents, no polynomial-time algorithm can achieve  an approximation ratio strictly below $\frac{5}{2}$ for the optimal EF contract unless $\P=\NP$ (\cref{hardness_ef}). Exact optimization remains \NP-hard with two agents (\cref{prop:two_ef_hard}). Instead, for the relaxed notions $\epsilon$-EF and EF1, we obtain additive FPTAS with respect to the optimal EF benchmark. The algorithms use dynamic programming over discretized utility profiles; the EF1 case requires an adaptive grid that preserves the EF1 constraints after discretization (\cref{thm:constant_agent,thm:constant_agentEF1}).

\vspace{1.5mm}

\noindent{\bf Price of fairness.}
We also quantify the revenue loss imposed by fairness constraints, using $\opt$ as the unconstrained benchmark. Exact EF can have unbounded price of fairness, already with two agents and one task (\cref{price_fair}).
Both relaxations of $\eps$-EF and EF1 restore the boundedness of the price of fairness. For $\eps$-EF, when $0<\eps\le 1/4$, we show that the principal can always retain at least a $4\eps/\min\{m,n^2\}$ fraction of $\opt$; conversely, single-task instances show that no guarantee strictly above a $4\eps$ fraction is possible in general (\cref{prop:pof_eps_ef}). Hence, for fixed $n$ and $m$, the worst-case price of fairness scales as $\Theta(1/\eps)$ as $\eps$ becomes small, and the single-task bound is tight up to an arbitrarily small slack. For EF1, we prove that the principal can always retain a $1/n^2$ fraction of $\opt$, while there are instances in which the retained fraction is only $O(1/\sqrt n)$ (\cref{ef1upperboundpro}).

\vspace{1.5mm}

\noindent{\bf Technical challenges and contributions.}
We close the overview by highlighting the technical ideas behind these results. The main difficulty is that envy-freeness is a global constraint on a mixed discrete-continuous object. The principal must choose both a task assignment and task-level shares, while every ordered pair of agents induces a comparison between two bundles evaluated through the first agent's own probabilities and costs. This prevents a separation between allocation and payment design: locally profitable contracts may create envy, and a fair-looking allocation may require revenue-destroying payments. Our general hardness reduction turns this coupling into an independent-set structure. High revenue requires assigning many vertex tasks to one efficient agent, while the envy constraints rule out adjacent choices; carefully scaled auxiliary tasks make the gap robust even when a constant amount of additive envy is allowed.

The algorithmic results identify regimes where this coupling can still be controlled. With a constant number of tasks, we can enumerate assignments; EF and $\epsilon$-EF then reduce to linear programs. EF1 is the first point at which this LP approach breaks: every ordered pair of agents carries an existential certificate, namely the task whose removal eliminates envy, and this certificate depends on the shares chosen by the LP. Naively enumerating these certificates is exponential because many agents may receive the empty bundle and may all compare themselves to the same nonempty bundle. Our key observation is that, for an empty-bundle agent, EF1 is equivalent to saying that she obtains positive utility from at most one task in the compared bundle. This converts the choice of removable tasks into a collection of break-even upper bounds on task shares. We enumerate threshold patterns induced by these break-even shares, discard patterns in which the same empty-bundle agent can be above threshold on two tasks, and solve the resulting LPs. Thus the existential EF1 witnesses are handled without enumerating them agent by agent.

The approximation algorithms for a constant number of agents require a different idea. We use dynamic programming over discretized $n^2$-dimensional utility profiles, recording how every agent values every other agent's contract. For $\epsilon$-EF, a uniform additive discretization is enough, because the target guarantee itself permits additive envy. For EF1, however, the final output must satisfy an exact ``up to one task'' condition; a small additive error can be spread over many tasks and need not be removable by deleting a single task. The algorithm therefore uses an adaptive grid rather than a fixed global mesh. It first guesses each agent's utility in the optimal EF benchmark and then chooses the granularity of the grid separately for that agent's utility scale. In this way, an agent with small benchmark utility is rounded more carefully, while an agent with large benchmark utility can tolerate coarser absolute rounding. The resulting errors behave like controlled relative losses in the envy comparisons, rather than unrelated additive perturbations. The dynamic program then finds a contract with the same discretized cross-utility profile as a rounded optimal benchmark solution, and the relative slack is small enough that deleting the largest task in the compared bundle absorbs the remaining error. This is the main mechanism that turns an additive dynamic-programming approximation into an exact EF1 contract with revenue within an additive $\epsilon$ of the EF benchmark.

The price-of-fairness results rely on separate, more instance-specific constructions. The unbounded EF example shows that the cost of exact fairness is intrinsic: with just one task and two agents, envy constraints can already preclude the unconstrained revenue-maximizing assignment. For $\epsilon$-EF, the positive bound combines two ways of spending the allowed envy: distributing a small envy budget task by task, and using a round-robin construction over single-task contracts. The matching dependence on $\epsilon$ comes from a single-task lower-bound instance. For EF1, the upper bound again uses a round-robin allocation, now centered on an agent with large total achievable surplus, while the lower bound adapts ideas from price-of-fairness constructions in indivisible goods but must encode them through success probabilities, costs, and incentive shares. Overall, the price-of-fairness results draw a qualitative boundary between exact and relaxed fairness: exact EF admits no finite worst-case guarantee, $\epsilon$-EF has tight $\Theta(1/\epsilon)$ dependence for fixed $n$ and $m$, and EF1 yields polynomial revenue guarantees while leaving the optimal dependence on $n$ as an open problem.

\subsection{Related Work}
The present paper lies at the intersection of algorithmic contract theory and fair division. Existing contract-design models optimize incentives and revenue, typically without agent-side fairness over assigned bundles; fair-division models study envy and efficiency, typically without hidden actions or outcome-contingent incentives.

\paragraph{Fairness-related considerations in contracts.}
Several strands of prior work incorporate fairness-related considerations into contract design. In the economics literature, \citet{fehr2007fairness} experimentally study how fairness concerns affect contractual relationships, showing that bonus contracts may outperform standard incentive contracts when agents exhibit fairness concerns. More recent work studies equal-pay-for-similar-work policies in labor markets~\citep{gentile2026equal}. In the algorithmic literature, \citet{feng2024price} study price-of-non-discrimination questions in public combinatorial contracts, where agents perform identical tasks. These equal-pay or non-discrimination constraints address exogenous restrictions, potentially imposed by regulation, rather than agents' endogenous fairness comparisons.
Such constraints are conceptually related to our approach, but they do not by themselves guarantee fairness. Agents ultimately care about their own subjective utilities in a mechanism, not only about the nominal payments they receive. In settings with heterogeneous agents and heterogeneous tasks, the same payment can generate different utilities, and hence envy among them, for different agents working on different tasks. See \cref{apx:equal_pay_not_fair} for an illustrative example. Thus, simply equalizing payments does not fully capture the subjective nature of fairness~\citep{yaari1984dividing}. In contrast, our envy-free notion is defined directly from the agents' subjective perspective: each agent compares her own utility with the utility she would obtain from another agent's assigned bundle of tasks and contracts, while exerting optimal effort.

\paragraph{Follow-up works on fair contracts.}
Following the initial version of the present paper, a recent line of work has studied related notions of fair contracts. The closest follow-up work is fair team contracts~\citep{castiglioni2025fair}, which studies fairness among agents selected for a single collaborative project. That paper also considers envy-freeness and extends the notion to environments with externalities. \citet{feldman2026equal,ding2026multi} study equal-pay contracts in team models, which constitute a special class of fair contracts under the definition of \citet{castiglioni2025fair}. These papers focus on team-production settings, whereas the present paper studies multiple individual tasks: each task is assigned to one agent, agents may receive bundles of tasks, and envy is evaluated over both the task allocation and the attached task-level contracts.

\paragraph{Algorithmic contract theory and combinatorial contracts.}
Contract design has attracted growing interest from the computer science community, where algorithmic methods enable the analysis of richer and more complex settings.
The algorithmic study of contracts was initiated by~\citet{babaioff2006combinatorial}. Subsequent work explored both the complexity of optimal contracts~\citep{dutting2021complexity} and approximation guarantees for simple linear contracts~\citep{dutting2019simple}. A line of recent work integrates moral hazard with adverse selection, providing algorithmic solutions in both single-parameter~\citep{alon2021contracts,alon2022bayesian} and multi-parameter settings~\citep{guruganesh2021contracts,castiglioni2021bayesian,castiglioni2022designing,guruganesh2023power,castiglioni2025reduction}.
Other recent directions include learning and online contract design~\citep{bacchiocchi2023learning,zhu2023sample,han2024learning}, agent-designed contracts~\citep{bernasconi2024agent}, budget-feasible contracts~\citep{feldman2025budget}, and contract design jointly with information structure or transparency~\citep{castiglioni2025hiring,dutting2026transparency}. We refer readers to the survey by~\citet{dutting2024algorithmic} for a comprehensive overview.

This work is closest to the literature on combinatorial and multi-agent contracts~\citep[e.g.,][]{dutting2023multi,dutting2022combinatorial,dutting2024combinatorial,deo2024supermodular,ezra2023approximability,duetting2025multi}. These papers introduce substantial combinatorial structure into contract design, but the objective is primarily revenue maximization rather than revenue maximization under agent-side envy constraints over bundles of tasks and contracts. The multi-project contract model of~\citet{alon2025multiprojectcontracts} is especially close because it also studies the delegation of multiple projects to heterogeneous agents. The allocation structure is different: in their model, each agent has limited capacity and can participate in at most one project, whereas here each task must be assigned to exactly one agent and an agent may receive multiple tasks. This changes both the combinatorial allocation problem and the fairness constraints, since envy is evaluated over bundles of tasks together with their attached contracts.

\paragraph{Fair division and envy-free pricing.}
The paper also relates to the literature on fair division under envy-free constraints. Envy-freeness is a central fairness notion, dating back to~\citet{foley1966resource}; see~\citet{aziz2022fair,amanatidis2023fair} for recent surveys on the fair allocation of indivisible goods and chores. Several well-known relaxations and variants have been proposed, including envy-freeness up to one item (EF1)~\citep{lipton2004approximately,budish2011combinatorial}, envy-freeness up to any item (EFX)~\citep{caragiannis2019unreasonable}, weighted envy-freeness~\citep{chakraborty2021weighted}, and proportionality~\citep{steinhaus1948problem,suksompong2016asymptotic}. The model here is also conceptually related to the fair division of chores~\citep[e.g.,][]{heydrich2015dividing,dehghani2018envy}.
The main distinction is incentive compatibility. In fair division, an allocated item or chore directly gives utility or disutility; in this setting, a task yields utility to an agent only through an outcome-contingent contract that makes effort worthwhile. This hidden-action feature creates a three-way interaction among assignment, incentives, and fairness. Moreover, envy-free full-allocation contracts always exist in this model, so the main challenge is revenue optimization under envy-freeness rather than existence.

The model is also related to envy-free pricing~\citep[e.g.,][]{guruswami2005profit,chen2010envy,anshelevich2017envy}, where a seller chooses prices and allocations so that buyers prefer their purchased bundles. In envy-free pricing, prices enter buyers' utilities symmetrically and the seller does not need to induce costly hidden effort. In fair contracts, the same task-level contract affects agents differently through heterogeneous costs and success probabilities, and the principal must choose both the assignment and the incentives.

\paragraph{Price of fairness.}
Finally, the price-of-fairness results are related to the literature quantifying the efficiency loss imposed by fairness constraints in goods allocation~\citep{bei2021price,barman2020optimal,bu2022complexity}. For example, \citet{barman2020optimal} obtain tight $\Theta(\sqrt{n})$ bounds for the price of EF1 in indivisible goods allocation. In contrast to fair-division settings where utilities are fixed by the allocation, the price of fairness in this paper is shaped by both the assignment and the incentive payments needed to induce effort.
For a broader overview of related results in fair division, we refer interested readers to the survey by \citet{liu2024mixed}.

\section{Model}\label{section_model_part}
\paragraph{Setting.} A principal has a set of tasks $\mathcal{M}$ to complete and can delegate them to a set of agents $\mathcal{N}$. Let $m=|\mathcal{M}|$ and $n=|\mathcal{N}|$ be some finite values. In particular, the principal has to assign each task $j\in \mathcal{M}$ to an agent $i \in \mathcal{N}$.
If task $j \in \mathcal{M}$ is assigned to agent $i$, then  agent $i \in \mathcal{N}$ chooses a binary action: exert effort, which costs $c_{i,j}\in[0,1]$ and succeeds with probability $p_{i,j}\in[0,1]$, or shirk, which costs zero and never succeeds.  If effort
succeeds, the principal receives a reward $r_j \in [0,1]$;
otherwise, the principal receives {\it zero}.  Agents are risk-neutral, effort choices are separable across tasks and utilities are additive across assigned tasks. We adopt principal-favoring tie-breaking: whenever effort and shirking give the same expected utility on an assigned task, the agent exerts effort.

To incentivize the agents to work, the principal designs a contract for each task. In this paper, we focus on the simple but widely adopted class of linear contracts. Specifically, a linear contract determines the fraction $\alpha_j\in [0,1]$ of the reward given to the agent to which the task is assigned. Under the linear contract $\alpha_j$, if an agent $i$ successfully completes task $j$, it will receive payment $\alpha_j \cdot r_j$; otherwise, it receives {\it zero} payments. Note that we use task-specific
linear shares $\alpha_j$ rather than agent-specific shares $\alpha_{i,j}$. Thus, if a worker compares herself to another
worker, she evaluates the same task-level contracts attached to the other worker’s bundle. This captures settings where platforms post task-level success-contingent payments and then assign
tasks to workers.

We use the shorthand $q_{i,j}=p_{i,j}r_j$ for the expected task reward generated by agent $i$ on task $j$, and
\[
u_{i,j}(\alpha_j)=\max\{\alpha_j q_{i,j}-c_{i,j},0\}
\]
for the utility agent $i$ obtains from task $j$ under share $\alpha_j$ when she may optimally shirk. For a bundle $S_h$, let $U_i(S_h;\alpha)=\sum_{k\in S_h}u_{i,k}(\alpha_k)$. We also define the break-even share
\[
\tau_{i,j}=
\begin{cases}
c_{i,j}/q_{i,j}, & q_{i,j}>0,\\
0, & q_{i,j}=c_{i,j}=0,\\
+\infty, & q_{i,j}=0<c_{i,j}.
\end{cases}
\]
Thus task $j$ can be effort-feasibly assigned to agent $i$ under some linear share in $[0,1]$ exactly when $\tau_{i,j}\le 1$.

\paragraph{Fair Contracts.} The main question we examine in this paper is how to design allocations of tasks to agents that satisfy fairness constraints.
We focus on cases in which each task can be assigned to at most one agent, while each agent may receive the delegation of multiple tasks. We use $S=\{S_i\}_{i\in \mathcal{N}}$ to denote the allocation, where $S_i \subseteq \mathcal{M}$ is the set of tasks assigned to agent $i$. Then, allocation $S$ is feasible if $S_i\cap S_j =\emptyset$ for all $i\neq j$. A contract $(S,\alpha)$ in our setting consists of a feasible allocation $S$ and the linear contracts $\alpha_j$ for each task $j$.
At this point, a feasible allocation may be partial. We call a contract effort-feasible if every assigned task satisfies the effort constraint below; it is full if every task is assigned.

Note that agents are heterogeneous in their efforts and success probabilities. To align with the {\it subjective} nature of fairness~\citep{yaari1984dividing}, we take envy-freeness as a starting point for our problem: each agent compares the utility they obtain from their own assigned bundle of tasks with the utility they would obtain from other agents' bundles.
This notion has been extensively studied in the fair division literature, and we build on it to define the envy-free (EF) contracts.
\begin{definition}[Envy-free Contracts]
A contract $(S,\alpha)$ is envy-free if
\begin{equation}\label{tmp_envy_freeness}
     \sum_{k \in S_i}\max\{\alpha_k p_{i, k }r_k - c_{i,k}, 0\} \geq \sum_{k \in S_j} \max\{\alpha_k p_{i, k} r_k - c_{i, k},0\},  \quad \forall j \neq i.
\end{equation}
\end{definition}
Intuitively, given the contract $(S,\alpha)$, envy-free contracts require that each agent $i$ prefers its assigned set of tasks over switching to another agent's set of tasks and exerting optimal effort. Equivalently, $U_i(S_i;\alpha)\ge U_i(S_j;\alpha)$ for all ordered pairs $i\neq j$.

To align with the literature on fair division, where all items must be assigned to one of the agents, we consider a setting in which the principal is restricted to full allocation contracts. That is, every task must be assigned to some agent, ensuring that no tasks are left unallocated. Moreover, we require that agents always exert effort to try to complete the tasks.\footnote{Notice that without this further restriction, assigning a task to an agent that does not exert effort is essentially equivalent to not assigning the task at all.}

\begin{definition}[Full Allocation Contracts]
\label{full_allocation_def}
A contract $(S,\alpha)$ is a full allocation contract if $\cup_{i\in \mathcal{N}} S_i = \mathcal{M}$ and
\begin{equation}\label{effort-constraints}
\alpha_j p_{i, j }r_j - c_{i,j} \geq 0, \quad \forall i \in \mathcal{N}, j\in S_i. \tag{Effort constraints}
\end{equation}
\end{definition}
Conceptually, these are effort constraints. Since shirking yields zero utility, \eqref{effort-constraints}
ensures that exerting effort weakly dominates shirking on every assigned task. By the tie-breaking convention above, weak dominance is sufficient for the principal's revenue expression to treat assigned agents as exerting effort.
Throughout our paper, we maintain the restriction to full allocation contracts.
Under this restriction, the envy-free constraints (\ref{tmp_envy_freeness}) can be simplified as follows:
\begin{equation}\label{envy-free-constraints}
     \sum_{k \in S_i} \alpha_k p_{i, k }r_k - c_{i,k} \geq \sum_{k \in S_j} \max\{\alpha_k p_{i, k} r_k - c_{i, k},0\},  \quad \forall j \neq i. \tag{EF constraints}
\end{equation}

We make the following assumption to ensure the existence of full allocation contracts.
\begin{assumption}
\label{asp:usefulness}
For any task $j\in \mathcal{M}$, there exists an agent $i\in \mathcal{N}$ such that $p_{i, j }r_j - c_{i,j}\geq 0$.
\end{assumption}
\noindent We argue that \cref{asp:usefulness} is rather mild: if the social welfare of a task is negative for every possible agent that performs the task, no agent would exert effort on it under any linear contract, and \eqref{effort-constraints} would necessarily fail. Hence,
the principal has to  remove these tasks in a preprocessing step before imposing full allocation on the remaining task set, which is indeed optimal.
Moreover, Assumption \ref{asp:usefulness} is reasonable in practice. For example, on a ride-hailing platform, if a customer sending a request is located in a too rural region or the payment is too low, the platform will have no incentives to consider that request and instead leave the customer unmatched.

\paragraph{Existence of EF Contracts.} Before moving to optimization, we first record that the fairness constraint is not an obstacle to feasibility. In sharp contrast to the fair allocation literature, where envy-free allocations may not exist, envy-free full-allocation contracts always exist in our model under \cref{asp:usefulness}.
\begin{proposition}\label{prop:ef_exists}
 Under \cref{asp:usefulness}, there always exists a (full allocation) contract that is envy free, and there exists an algorithm that finds it in polynomial time.
\end{proposition}
The construction is the zero-rent one: assign each task $j$ to an agent with minimum break-even share $\tau_{i,j}$ and set $\alpha_j=\tau_{i,j}$. The assignee obtains zero utility, and every other agent weakly prefers shirking on that task, so all agents value all bundles at zero. Thus existence is easy because EF is a relative comparison; the hard problem is revenue optimization under EF.

The next observation explains why the full-allocation restriction does not lose revenue under the same usefulness assumption: any unallocated useful task can be added at its break-even share without creating envy.
\begin{proposition}\label{prop_full_allo}
Under \cref{asp:usefulness}, for any envy-free partial allocation contract $(S,\alpha)$, there exists an envy-free full allocation contract $(\hat{S},\hat{\alpha})$ that generates a weakly higher revenue for the principal.
\end{proposition}

 Therefore, from now on, we will refer to the {\it full allocation contracts} simply as {\it contracts}.

\paragraph{The Principal's Problem.} Given any contract $(S,\alpha)$, the expected revenue of the principal is
\begin{align*}
\rev(S,\alpha)=\sum_{i\in \mathcal{N}} \sum_{k \in S_i} (1-\alpha_{k}) p_{i,k}r_k.
\end{align*}
In this paper, the objective of the principal is to design a fair contract that maximizes its expected revenue, subject to the EF constraints and the effort constraints. Specifically, we have the following optimization program.
\begin{align} \label{pr:fair}
\max_{\{\alpha_k\}_{k\in \mathcal{M}}, \{S_i\}_{i\in \mathcal{N}}} \quad & \rev(S,\alpha) \tag{\optval}\\
\text{subject to} \qquad & \text{(\ref{effort-constraints}), (\ref{envy-free-constraints}) } \nonumber\\
&S_{i} \cap S_j =\emptyset, \quad \forall i\neq j \nonumber\\
& \bigcup_{i\in \mathcal{N}} S_i = \mathcal{M}. \nonumber
\end{align}
We also use $\optval$ to denote the optimal objective value for program \eqref{pr:fair}.

\paragraph{Relaxed Notions of EF contracts.} We also consider two relaxations of EF contracts: $\eps$-EF contracts and EF1 contracts. $\eps$-EF contracts relax the EF constraints by some $\eps>0$, while EF1 contracts relax EF to envy-free up to one task. Note that Proposition \ref{prop_full_allo} and \ref{prop:ef_exists} also hold for $\eps$-EF contracts and EF1 contracts.
\begin{definition}[$\eps$-Envy-Free Contracts]
\label{def:eps_ef}
A contract $(S,\alpha)$ is $\eps$-envy-free if
\begin{equation}
     \sum_{k \in S_i} \alpha_k p_{i, k }r_k - c_{i,k} \geq \sum_{k \in S_j} \max\{\alpha_k p_{i, k} r_k - c_{i, k},0\} - \eps,  \quad \forall j \neq i. \tag{$\eps$-EF constraints}
\end{equation}
\end{definition}

\begin{definition}[Envy-Free-Up-to-One (EF1) Contracts]
\label{def:ef1}
A contract $(S,\alpha)$ is envy-free up to one task (EF1) if for every ordered pair of agents $i, j$ with $i\neq j$,  either $S_j=\emptyset$, or there exists a task $k_{i,j} \in S_j$ such that
\begin{equation}
     \sum_{k \in S_i} \alpha_k p_{i, k }r_k - c_{i,k} \geq \sum_{k \in S_j\backslash\{k_{i,j}\}} \max\{\alpha_k p_{i, k} r_k - c_{i, k},0\},  \quad \forall j \neq i. \tag{EF1 constraints}
\end{equation}
\end{definition}

Let $\opteps$ and $\optefw$ be the optimal objective values by replacing the fairness constraints in the optimization program \eqref{pr:fair} with $\epsilon$-EF and EF1 respectively.

\paragraph{Additional notations.} In the remaining sections of the paper, we let $\lceil x\rceil_{\mathcal{X}}$ be the smallest element in $\mathcal{X}$ greater than or equal to $x$, given a $x\in \mathbb{R}$ and a discrete set $\mathcal{X}$. Also, let $[n] \triangleq \{1, 2, \dots, n\}$ for any integer $n$, and $\lfloor x \rfloor$ be the largest integer smaller than or equal to $x$.

\section{Computational Hardness for EF Contracts}
\label{sec:general}
In this section, we show that despite the desirable property that envy-free contracts always exist (\cref{prop:ef_exists}), finding a revenue-optimal one is computationally challenging.
Specifically, we show that even finding an approximately optimal envy-free contract  is \NP-hard for any constant approximation ratio (\cref{arbitryri_hardnesdd_linear}).
Moreover, we show that even if we restrict our attention to settings with a constant number of agents, beating the approximation ratio $\frac{5}{2}$ is \NP-hard (\cref{hardness_ef}).
The missing proofs from this section are provided in \cref{missingproof_sec_gener}.

\subsection{Hardness for General Settings}
In this section, we prove the hardness results of computing optimal EF contracts. In particular, we provide a stronger result that it is \NP-hard to approximate optimal EF contracts within any constant approximation ratio, even using the relaxed $\eps$-EF contracts for approximation.
The proof is by a reduction from the following problem.

\begin{definition}[$\textsf{GAP-BOUNDED-IS}_{\mu,\eta}$] For every $\mu \in [0,1]$ and $k \in \mathbb{N}_+$, we define $\textsf{GAP-Bounded-IS}_{\mu,k}$ as the following promise problem.
Given a graph $G=(V,E)$ in which each vertex has degree at most $k$ and an $\eta\in [\frac{1}{k},1]$ such that either one of the following is true:

\begin{itemize}
	\item there exists an independent set of size at least $\eta |V|$,
	\item all the independent sets have size at most $\mu \eta |V|$,
\end{itemize}
determine which of the two cases holds.
\end{definition}
It is well known that for every constant $\mu>0$, there exists a constant $k= k(\mu)$ that depends on $\mu$ such that $\textsf{GAP-Bounded-IS}_{\mu, \eta}$ is \NP-hard \citep{alon1995derandomized,trevisan2001non}.

\begin{theorem}\label{arbitryri_hardnesdd_linear}
For any constant $c \ge 1$, there exists a constant $\epsilon>0$ such that, assuming {\rm \P~$\neq$~\NP}, there does not exist a polynomial-time algorithm that computes an $\epsilon$-EF contract with revenue at least a $\frac{1}{c}$ fraction of $\optval$.
\end{theorem}

\begin{proof}

The theorem is proved by reducing from the problem $\textsf{GAP-Bounded-IS}_{\mu, \eta}$.
For any desired constant approximation ratio $c$, we reduce from the problem $\textsf{GAP-Bounded-IS}_{\mu, \eta}$ with constants $\mu=\frac{1}{2c}$ and $k=k(\mu)$.
	In the following, we let $\gamma=|V|$, $\beta=|E|$, and we define the constant $\delta=\frac{1}{8ck^2}$, and $\epsilon = \delta/3$.
	Notice that since the maximum degree is $k$, we have $\beta\le \frac{\gamma k}{2} $.

	\paragraph{Construction.} Given a graph $G=(V,E)$, we construct an instance as follows.
	We first define the set of agents. For each edge $e\in E$, there exists an agent $b_{e}$. Additionally, we have one ``efficient" agent $b^*$.
	Then, we define two sets of tasks:
	\begin{itemize}
		\item For each vertex $v\in V$, we have a task $d_v$ with reward $r_{d_v}=1$. The probabilities of success and costs are constructed as follows.
		\begin{itemize}
			\item Let $p_{b^*,d_v}=1$ and $c_{b^*,d_v}=1/2$.
			\item If node $v\in e$, i.e.,  the edge $e$ connects $v$ to another node, then let  $p_{b_{e},d_v}=\delta$ and $c_{b_{e},d_v}=0$.
			\item Let $p_{b,d_v}=0$ and $c_{b,d_v}=0$  for all other agents $b$.
		\end{itemize}
		\item For each edge $e\in E$, we define a task $ d_{e}$ with $r_{d_{e}}=\delta/2$.  The probabilities of success and costs are defined as follows.
		\begin{itemize}
			\item For agent $b_{e}$, let $p_{b_{e}, d_{e}}=1$ and $c_{b_{e}, d_{e}}=0$.
			\item Let $p_{b,d_{e}}=c_{b,d_{e}}=0$ for all other agents $b$.
		\end{itemize}
	\end{itemize}

\paragraph{Key idea of the proof.}  To prove the statement, we mainly show the following two key statements: (i) If there exists an independent set of size at least $\eta |V|$, then there exists an EF contract with principal's expected revenue at least $\frac{1}{2} \eta\gamma$; (ii)
If all independent sets have size at most $\mu \eta \gamma$, then all $\epsilon$-EF contracts achieve principal's expected revenue strictly less than $\frac{1}{2c}\eta\gamma$.

Intuitively, to achieve a high revenue, the principal should assign as many tasks as possible to the efficient agent $b^*$, with a contract at least $\frac{1}{2}$. Therefore, if there exists an independent set of size at least $\eta |V|$, we can construct an EF contract that assigns at least $\eta |V|$ tasks to $b^*$, yielding a revenue of  $\frac{1}{2} \eta\gamma$. Otherwise, if all independent sets have size at most $\mu \eta |V|$, the key for proving the upper bound on revenue is to show that the number of tasks assigned to agent $b^*$ is at most $\mu\eta\gamma$, even if we allow $\eps$-EF constraints, where $\eps>0$ is a small constant and $\eps<\frac{\delta}{2}$. This is shown by contradiction: There must be tasks $d_v$ and $d_{v'}$ assigned to $b^*$, where $(v,v')=e \in E$, and the contract on each of these tasks is at least $\frac{1}{2}$ by the effort constraints. Hence, agent $b_{e}$ would get $\delta$ performing tasks $d_v$ and $d_v'$. Thus, it must be the case that the payment to agent $b_{e}$ from $d_{e}$ (the only other task with positive success probability) is $\delta$, which is impossible since $r_{d_{e}}=\frac{\delta}{2}$ and violates the EF constraint by at least $\frac{\delta}{2} > \eps$.

	\paragraph{Sufficiency.}

	Assume that there is an independent set $I$ such that $|I|\ge \eta \gamma$. 
		Consider the set of tasks $\{d_v\}_{v\in V}$. We assign all the tasks $d_v$ with $v\in I$ to agent $b^*$ and we set contract $\alpha_{d_{v}} = \frac{1}{2}$ for every $v\in I$. The effort constraints for agent $b^*$ are therefore satisfied.
	We assign every remaining task $d_v$, where $v\in V\setminus I$, to any agent $b_{e^*}$ with $v \in e^*$. We set the contract for these tasks $d_v$ to $\alpha_{d_v} = 0$. Note that this will not increase the utility of agent $b_{e^*}$, but it will increase the principal's revenue by at most $(\gamma - |I|)\delta$.

	Next, we consider the set of tasks $\{d_{e}\}_{e\in E}$. For any edge $e\in E$, note that at most one of its nodes $v\in e$ belongs to the set $I$.
	If such a node $v\in e\cap I$ exists, then agent $b_{e}$ would receive utility $\frac{1}{2}\delta$ by switching to the set of tasks $I$ assigned to $b^*$. Hence, we assign the task $d_{e}$ to this agent $b_{e}$ by setting contract $\alpha_{d_{e}}=1$, which leads to satisfying EF constraints.
	Otherwise, i.e., if $e\cap I=\emptyset$, the corresponding agent $b_{e}$ gets utility $0$ by switching to agent $b^*$'s assigned set $I$. Hence, we can assign the task $d_{e}$ to agent $b^*$ by setting contract $\alpha_{d_{e}} = 0$. Note that this will not change agents' utilities, and the EF constraints continue to hold. The principal gets revenue $0$ for this set of tasks.

	Hence, we conclude that the principal's revenue is
	\begin{equation*}
		\frac{1}{2}|I| + (\gamma - |I|)\delta \ge \frac{1}{2}\eta \gamma.
	\end{equation*}

	\paragraph{Necessity.}
	Assume that all the independent sets have size at most $\gamma \eta \mu$.
		Consider an allocation $S$ and contracts $\alpha$. Let $S_{b^*}$ be the set of tasks assigned to $b^*$ and $V^*$ be the subset of nodes in $V$ whose corresponding tasks are assigned to agent $b^*$, i.e., $V^* = \{v \in V\mid d_v \in S_{b^*} \}$. By the effort constraints, it must be the case that $\alpha_{d_v} \ge \frac{1}{2}$ for all $v\in V^*$.

	We first show that the size $|V^*| \le \gamma \eta \mu$. The following argument is by contradiction. Suppose $|V^*| > \gamma \eta \mu$. By assumption, we have that $V^*$ is not an independent set. Therefore, there must exist some edge $e = (v_1, v_2)$ such that both nodes $v_1, v_2 \in V^*$. This implies that both tasks $d_{v_1}$ and $d_{v_2}$ are assigned to agent $b^*$. Hence, by the EF constraints, if an agent $b_{e}$ switches to work on the set of tasks $S_{b^*}$, it would receive payments at least
	\begin{align}\label{eq:cont1}
		\alpha_{d_{v_1}}p_{b_{e}, d_{v_1}} r_{d_{v_1}} + \alpha_{d_{v_2}}p_{b_{e}, d_{v_2}} r_{d_{v_2}} \ge \delta.
	\end{align}
	Hence, the payment to agent $b_{e}$ is at least $\delta$. However, since we are restricted to linear contracts, the payments to agent $b_{e}$ is upperbounded by $\delta/2$. Indeed,  the only task outside $S_{b^*}$ from which $b_e$ can obtain positive utility is $d_e$, whose payment is at most $\delta/2$.
	Hence, the EF constraint is violated by at least $\delta/2>\epsilon$ by Equation \eqref{eq:cont1}, reaching a contradiction.

	Thus, we can restrict to allocations in which $|V^*|\le \gamma \eta \mu$.
	In this case, the principal's revenue is at most
	\[
	\frac{1}{2}|V^*| + \delta(\gamma - |V^*|)+ \frac{\delta}{2}\beta \le 	\frac{1}{2} \eta \mu \gamma + \delta (\gamma+ \frac{\gamma k}{2})< \frac{1}{2c} \eta \gamma
	\]
    where in the first inequality we use $\beta\le  \frac{\gamma k}{2}$, and the strict inequality follows from $\mu=1/(2c)$, $\delta=1/(8ck^2)$, $k\ge 1$, and $\eta\ge 1/k$.
	This concludes the proof.
\end{proof}

This immediately implies that there is no PTAS for computing  EF-optimal contracts.
\begin{corollary}
    	For any constant $c \ge 1$, assuming {\rm \P~$\neq$~\NP}, there does not exist a polynomial-time algorithm that computes an EF contract with revenue at least a $\frac{1}{c}$ fraction of $\optval$.
\end{corollary}

\subsection{Hardness for Settings with A Constant Number of Agents}
\label{sec:constant_agent}
In this section, we examine the hardness of computing optimal contracts in settings with a constant number of agents.
We first show that even in the restricted setting with a constant number of agents,
computing an EF contract with an approximation ratio  strictly below $\frac{5}{2}$ is still \NP-hard.
This rules out the possibility of a multiplicative PTAS for EF contracts.

\begin{theorem}\label{hardness_ef}
When there are three agents, assuming {\rm \P~$\neq$~\NP}, for every constant $\rho>0$ there does not exist a polynomial-time algorithm that computes an envy-free contract with revenue at least a $\left(\frac{2}{5}+\rho\right)$ fraction of $\optval$.
\end{theorem} 

Interestingly, the hardness proof of \cref{hardness_ef} can be easily adapted to show that it is \NP-hard to compute a constant-ratio approximation of optimal EF1 and $\eps$-EF contracts.

\begin{corollary}\label{hardness_ef1_epsef}
    When there are three agents, assuming {\rm \P~$\neq$~\NP}, for every constant $\rho>0$ we have that (i) there does not exist a polynomial-time algorithm that computes an EF1 contract with revenue at least a $\left(\frac{7}{10}+\rho\right)$ fraction of  $\optefw$; and (ii) for any $\eps<\frac{1}{20}$, there does not exist a polynomial-time algorithm that computes an $\eps$-EF contract with revenue at least a $\left(\frac{7}{10}+\rho\right)$ fraction of  $\opteps$.
\end{corollary}

Moreover, in the more extreme case with only two agents, we show that computing the optimal EF contract remains \NP-hard. It remains open whether nontrivial approximation guarantees are possible
for the two-agent case.
\begin{proposition}
\label{prop:two_ef_hard}
When there are only two agents, assuming {\rm \P~$\neq$~\NP}, there does not exist a polynomial-time algorithm that computes the optimal envy-free contract.
\end{proposition}

\section{Polynomial-time Algorithms in Restricted Settings}\label{section_resutionced_alg}

In this section, we focus on the design of polynomial-time algorithms that may achieve ``good" revenue in restricted settings. Recall the strong negative results \Cref{arbitryri_hardnesdd_linear} that in general settings of an arbitrary number of agents and tasks, there is no PTAS for approximating the optimal EF contracts even with $\eps$-EF contracts. This motivates us to consider two restricted settings with (i) either a constant number of tasks, or (ii) a constant number of agents.

In terms of settings with a constant number of tasks, we show that the optimization problem can be solved via solving a set of linear programs with polynomial size, and hence the optimal EF-contracts, $\eps$-EF contracts, and EF1 contracts can be found in polynomial time (\cref{thm:constant_task}).

The second set of results studies settings with a constant number of agents. Since \Cref{hardness_ef} shows that the problem  does not admit a PTAS under exact EF unless \P=\NP,  we first circumvent such hardness results by considering two relaxations,  $\eps$-EF and EF1. Then, for each relaxation, we present an (additive) FPTAS via dynamic programming (\cref{thm:constant_agent,thm:constant_agentEF1}). We leave open the question of whether exact EF contracts admit any nontrivial constant approximation ratio when the number of agents is constant.

\subsection{Polynomial-time Algorithms for a Constant Number of Tasks}
\label{subsec:const_tasks}

We first present polynomial-time algorithms for computing optimal contracts under EF and EF1 constraints.

\begin{theorem}\label{thm:constant_task}
When the number of tasks $m$ is a constant, there exists a polynomial-time algorithm that computes optimal contracts in $\poly(n)$ time under EF, $\eps$-EF and EF1 constraints.
\end{theorem}

We first observe from Program (\ref{pr:fair}) that if the allocation of tasks is fixed, Program (\ref{pr:fair}) can be reformulated as a linear program. Therefore, if the number of tasks $m$ is a constant, we can enumerate all the possible allocations of tasks, where the number of such allocations is at most $O(n^m)$. For each allocation, we solve the corresponding LP. All these steps clearly take polynomial time. Similarly, we can apply a similar procedure to compute the optimal $\eps$-EF contracts.

The proof for computing optimal EF1 contracts is more involved. While the high-level idea is similar to EF-contract cases, which first enumerate the allocation of tasks and then formulate a polynomial-size linear program, the challenges come from the requirements that we need to additionally decide the task to be removed (called {\it removable task}) in EF1 constraints. A naive enumeration for removable tasks can result in an exponential number of possibilities. For example, suppose all tasks are assigned to agent $i$ in one enumeration. For any other agent $j\neq i$, to satisfy the EF1 constraints of agent $j$ envying agent $i$, one needs to determine the removable task. Therefore, simply enumerating the removable tasks for EF1 constraints of all other $j\neq i$ envying agent $i$ would result in $O(m^n)$ possibilities.

To overcome the above challenge, we enumerate the contracts adopted. Without loss of generality, consider the EF1 constraints of agent $i$ envying agent $j$, whose assigned set of tasks are $S_i$ and $S_j$, respectively. Since we have a constant number of tasks, the only case that leads to exponential complexity is $|S_i| = 0$ and $|S_j|>0$, as discussed above. Notice that agent $i$ receives utility $0$ from set $S_i$. To ensure EF1 constraints, we only need to ensure that agent $i$ gains positive utility from at most one task (e.g., task $k$) in $S_j$. This implies one possible upper bound on the contract $\alpha_k$, which is the {\it  break-even share} sufficient to incentivize agent $i$ to exert effort on task $k$. Hence, instead of enumerating the removable tasks, we enumerate the  break-even shares to ensure that every empty-bundle agent gains positive utility from at most one task in $S_j$. This reduces the exponential complexity to polynomial complexity $O(n^m)$.

\begin{proof}[Proof of \cref{thm:constant_task}] We present the proof for computing optimal EF contracts, $\eps$-EF contracts and  EF1 contracts separately.

\paragraph{Computing optimal EF contracts.}
Notice that if the number of tasks $m$ is a constant, there are at most $O(n^m)$ different allocations of tasks. To prove the result, it is sufficient to show that given an allocation $S_1, S_2, \dots, S_n$, Program (\ref{pr:fair}) is a linear program. By introducing additional variables $t_{i,k}$ for all $i \in \mathcal{N}, k\in \mathcal{M}$, we have the Program (\ref{pr:fair}) become
\begin{align} \label{OPT_EF_LP_add_var}
\max_{\{\alpha_k\}_{k\in \mathcal{M}}} \quad & \rev(S,\alpha) \tag{OPT-EF-LP}\\
\text{subject to} \qquad & \alpha_j p_{i, j }r_j - c_{i,j} \geq 0, \quad \forall i \in \mathcal{N}, j\in S_i. \nonumber\\
& \sum_{k \in S_i} \alpha_k p_{i, k }r_k - c_{i,k} \geq \sum_{k \in S_j} t_{i,k},  \quad \forall j \neq i. \nonumber \\
&t_{i, k} \ge \alpha_k p_{i, k} r_k - c_{i, k}, \quad \forall i \in \mathcal{N}, k\in \mathcal{M}. \nonumber \\
&t_{i, k} \ge 0, \quad \forall i \in \mathcal{N}, k\in \mathcal{M}. \nonumber
\end{align}
which is a linear program with polynomial number of variables and constraints. This concludes the proof.

\paragraph{Computing optimal $\eps$-EF contracts.} Following the proof for EF contracts, we can find the optimal $\eps$-EF contracts by first enumerating the allocation of tasks, and then formulating a linear program.

\paragraph{Computing optimal EF1 contracts.}
    The proof for EF1 contracts is more involved. We first note that given an allocation $S_1, S_2, \dots, S_n$, there are at most a constant number of agents assigned a non-empty set of tasks. The number of such allocations is at most $O(n^m)$. Hence, as in the proof of \cref{thm:constant_task}, we can enumerate over all the allocations.

    Hence, in the following we fix an allocation $S_1, S_2, \dots, S_n$. Consider the EF1 constraint relative to a pair of agents $(i, j)$, which guarantees that agent $i$ prefers $S_i$ over $S_j$ up to one item. Note that if we directly enumerate all the possible \emph{removable tasks} $k_{i,j}$ (as in Definition \ref{def:ef1}) for all pairs of agents $(i,j)$, there would be an exponential number of possible sets of constraints, \emph{i.e.}, combinations of removal tasks. Our main result is that we actually only need to enumerate a polynomial number of possible sets of constraints. Then, for each allocation and combination of removal tasks, we solve a linear program.

    For pair $(i, j)$, consider the following four cases:
    \begin{enumerate}[(i)]
    \item  $|S_i|>0$ and $|S_j|=0$
    \item $|S_i|=0$ and $|S_j|=0$
    \item $|S_i|>0$ and $|S_j|>0$
    \item $|S_i|=0$ and $|S_j|>0$
\end{enumerate}

(i) Since the effort constraints give agent $i$ weakly nonnegative utility from every assigned task, the EF1 constraints are satisfied.

(ii) Since both sides of EF1 constraints are $0$, the EF1 constraints are satisfied.

(iii)  There are at most $|S_j|$ different possible EF1 constraints, which are obtained by enumerating all tasks $k_{i,j}\in S_j$. Since there are at most $m$ agents such that $|S_i|>0$, we get that the total number of possible sets of EF1 constraints related to an agent $j$ is $O(m^m)$. Moreover, there are at most $m$ such $S_j$, and hence the total number of possible sets is $O((m^m)^m)$.

(iv) Suppose the set of agents that are assigned the empty set is $A$. First, note that there are at most $n$ agents that are assigned empty sets of tasks. Hence, we cannot directly enumerate all the removable tasks $k_{i,j}$ for each agent $i$. However, the EF1 constraints of pairs $(i,j)$ imply that
    \begin{equation}\label{aps_ef1_0removeal}
    0 \geq \sum_{k \in S_j\backslash\{k_{i,j}\}} \max\{\alpha_k p_{i, k} r_k - c_{i, k},0\},  \quad \forall j \neq i, \exists k_{i,j} \in S_j
\end{equation}
which implies that for task $k\in S_j\setminus \{k_{i,j}\}$, the contract must be at most
\[
\alpha_k \le \tau_{i,k}
\]
This means that if the removable task is $k_{i,j}$, as long as  $\alpha_k \le \tau_{i,k}$ holds for all tasks $k\in S_j\setminus \{k_{i,j}\}$, the agent pair $(i,j)$ is EF1.  Equivalently, an empty agent can have positive utility for at most one task in the compared bundle.
Moreover, in this case, the contract $\alpha_{k_{i,j}}$ of the removable task  $k_{i,j}$ has no restriction. In other words, there exists at most one task in $S_j$, from which agent $i$ can get positive utility.

Given an EF1 contract, suppose the set of removable tasks from $S_j$ is $\{k_{i,j}\}_{i\in A}$ for the set of agents $A$. Note that $\{k_{i,j}\}_{i\in A}$ is a multiset.
By the above discussion, it determines the upper bounds (i.e., the  break-even shares) for all contracts $\alpha_k$ of all tasks $k\in S_{j}$, which is also a sufficient condition for EF1 constraints to hold for all agents $i\in A$ regarding set $S_j$.
Hence, instead of enumerating removable tasks for all agents $i\in A$ regarding set $S_j$, we can equivalently enumerate the upper bounds for contracts of all tasks in $S_j$.

The algorithm for case (iv) works as in Algorithm \ref{algorithm_ef1polynomialcase4}.

\begin{algorithm}[t]
\caption{Find Upper Bounds for Contracts $\alpha_k$, $k\in S_j$ \\
\textbf{Input:} Set of tasks $S_j$,  set of empty-bundle agents $A$, $\{p_{i,j}\}_{i \in {A}, j \in S_j}$, $\{c_{i,j}\}_{i \in {A}, j \in {S_j}}$, $m, n$ \\
\textbf{Output:} All feasible sets of indexes $\{i_k\}_{k \in S_j}$.
}\label{algorithm_ef1polynomialcase4}
Compute the break-even shares $\tau_{i,k}$ for all agents $i\in A$ and tasks $k\in S_j$ \;
Sort the distinct finite values in $\{\tau_{i,k}\}_{i \in A}$ for task $k$ from low to high. Denote the list of threshold values by $l_k$. \;
$\mathcal{F} \gets \emptyset$ \;
\For{ every set of indexes $\{i_k\}_{k \in S_j}$}{
   \If{$\{i_k\}_{k \in S_j}$ is feasible}{
$\mathcal{F} = \mathcal{F} \cup \{\{i_k\}_{k \in S_j}\}$ \;
   }
}
\Return $ \mathcal{F}$.
\end{algorithm}

For each task $k \in S_j$, we calculate the break-even share $\tau_{i,k}$ for agent $i\in A$. Dispose $\tau_{i,k}$ if $\tau_{i,k}> 1$. To ease exposition, assume all break-even shares $\tau_{i,k}$ are at most $1$.  Then, for each task $k$, we sort all agents according to their break-even shares from low to high. Additionally, we add one {\it virtual} threshold value $\omega=1$ to serve as an upper bound on contracts. Denote this list as $l_k$.
Now,  we have $|S_j|$ such lists. To determine the upper bounds for contract $\alpha_k$ of task  $k\in S_j$, we can iterate over the threshold values in list $l_k$ and use them as candidate upper bounds.

Suppose in one iteration, we obtain indexes $\{i_k\}_{k\in S_j}$ for all the lists.
Then, given an allocation of sets $S_1, S_2, \dots, S_n$, which can be obtained by enumeration, we form a linear program
\begin{align} \label{ef1_tmp}
\max_{\{\alpha_k\}_{k\in \mathcal{M}}} \quad & \rev(S,\alpha) \tag{EF1-Program}\\
\text{subject to} \qquad & \text{(\ref{effort-constraints})} \nonumber\\
& \text{ EF1 constraints from case (iii) } \nonumber\\
& \text{ Upper bounds  $\alpha_k \le l_k [i_k] $ for all the tasks} \nonumber
\end{align}

However, not every enumeration of upper bounds is desired, as some of them may lead to a solution of (\ref{ef1_tmp}) that is not EF1. To see this, consider an agent $a \in A$ and two tasks $k, k' \in S_j$. The indexes of agent $i$ in lists $l_k, l_{k'}$ are $a_k, a_{k'}$, respectively. Suppose in one enumeration, we have enumerated indexes $i_k > a_k$ and $i_{k'}>a_{k'}$. By solving (\ref{ef1_tmp}), it may lead to one \emph{undesired} solution where $\tau_{l_k [a_k],k} < \alpha_{k} \le \tau_{l_k [i_k],k}$ and $\tau_{l_{k'} [a_{k'}],k'} < \alpha_{k'} \le \tau_{l_{k'} [i_{k'}],k'}$. This solution is not EF1 for agent $a$, as it gives agent $a$ a positive utility from exerting efforts on set $S_j$.

To rule out the above undesired enumerated upper bounds, we need to \emph{check the feasibility of enumerated upper bounds} as follows: If there do not exist two lists $l_k, l_{k'}$ where $k, k' \in S_j$ such that there is an agent $a\in A$ appearing in both lists $l_k[1, \dots, i_k-1]$ and $l_{k'} [1, \dots, i_{k'}-1]$ for two enumerated index $i_k, i_{k'}$, the enumeration of upper bounds is feasible. 

Finally, we show that there are at most a polynomial number of possible sets of constraints in case (iv).
Note that there are at most $O(n^{|S_j|})$ possible sets of upper bounds for the set $S_j$. Since there are at most $m$ tasks, the total number of possible sets of upper bounds is $O(n^m)$.

\medskip

In summary of the above four cases, given a task allocation, we need to solve at most $O(m^{m^2}  n^m)$ different programs (\ref{ef1_tmp}). For each task allocation, EF1 constraints, and upper bounds, we check if the enumerated upper bounds are feasible. If it is the case, we proceed to solve an LP:
For each program (\ref{ef1_tmp}), by introducing additional variables $t_{i,k}$ for all $i \in \mathcal{N}, k\in \mathcal{M}$ and rewriting  EF1 contracts as in (\ref{OPT_EF_LP_add_var}), we obtain a linear program with a polynomial number of variables and constraints.  Finally, since the number of allocations of tasks is at most $O(n^m)$, we have that the complexity of computing EF1 contracts is polynomial.
This concludes the proof.
\end{proof}

\subsection{FPTAS for a Constant Number of Agents}

In this section, we focus on polynomial-time algorithms for finding approximately optimal contracts (with respect to optimal EF contracts) under two relaxed notions of fairness, $\eps$-EF
 and EF1. We start by introducing a general dynamic programming template.

\paragraph{A General Dynamic Programming Template.}
An approximately optimal contract can be computed efficiently by discretizing the set of possible contracts and utilities.
In particular, we consider a setting where
\begin{itemize}
\item the contract $\hat \alpha_j$ for task $j \in \mathcal{M}$ belongs to a discrete set $\mathcal{D}_j$.
\item the utility $\hat U_{i,j}(\hat \alpha_j)$ of agent $i$ exerting {\it optimal} effort on task $j$ under contract $\hat \alpha_j\in\mathcal{D}_j$ belongs to a discrete set $\mathcal{U}_i$.
\item the principal's expected revenue $\hat U^P(\hat \alpha_j,i)$ from a task $j$ performed by agent $i$ under contract  $\hat \alpha_j$ belongs to a discrete set $\mathcal{H}$.
\end{itemize}

Algorithm \ref{algorithm_dynmaic_pr2} provides a general dynamic programming template.
The algorithm iterates over all the tasks $\ell \in \mathcal{M}$ and keeps track of the possible utility vectors that can be induced by contracts satisfying the effort constraints.
Formally, it keeps track of a set $\mathcal{L}_\ell$ of all the possible reachable profiles $(\hat S,\hat \alpha, y)$,
where considering assigning the first $\ell$ tasks,
\begin{itemize}
\item $\hat S$ is a partial allocation of the tasks, and $\hat \alpha$ represents the contracts of the tasks;
\item $y=(\hat{h}, (\hat v_{i,j})_{i,j \in \mathcal{N}})$ is the utility profile of agents and the principal:  $\hat{h}$ represents the principal's expected revenue,
and $\hat{v}_{i,j}$ is the utility of agent $i$ working on the tasks assigned to agent $j$.
\end{itemize}

The reachable profiles are computed as follows. For each contract $\hat \alpha_\ell \in \mathcal{D}_\ell$ and agent $k\in\mathcal{N}$, the algorithm checks whether assigning task $\ell$ to agent $k$ under $\hat \alpha_\ell$ satisfies the effort constraint. If so, it computes all the possible partial allocations obtained by assigning task $\ell$ to agent $k$ with contract $\hat \alpha_\ell$,  given the partial allocation $(\hat S,\hat \alpha,(\hat h,(\hat v_{i,j})_{i,j \in \mathcal{N}}))\in \mathcal{L}_{\ell-1}$ (Lines \ref{line:update1} $\sim$ \ref{line:update4} in Algorithm \ref{algorithm_dynmaic_pr2}). Specifically, the utility profile is computed as follows:
\begin{subequations} \label{eq:update}
\begin{align}
    &h'= \hat h + \hat U^P(\hat \alpha_\ell,k)\\
    & v'_{i,k}= \hat v_{i,k} + \hat U_{i,\ell}(\hat \alpha_\ell)\quad \forall i \in \mathcal{N}\\
    & v'_{i,j}= \hat v_{i,j}  \quad \forall i\in \mathcal{N}, \forall j\neq k
\end{align}
\end{subequations}
Note that our goal is only to have a representative contract (and allocation) for each reachable utility profile $y$. If a  pair $(\hat S',\hat \alpha')$ induces the same utility profile as another  pair in $\mathcal{L}_{\ell}$, it is removed (See Line \ref{line:remove}).

Finally, for each utility profile $y$, the algorithm returns a  pair $(\hat S,\hat \alpha)$ that induces it.
Formally,
\begin{proposition}\label{thm:DP}
Algorithm \ref{algorithm_dynmaic_pr2} returns a set $\mathcal{C}$ such that (i) all $(\hat S,\hat \alpha)\in \mathcal{C}$ satisfy the effort constraints, and (ii) for each $(S,\alpha)$ with $\alpha \in  \prod_{j \in \mathcal{M}}\mathcal{D}_j$ that satisfies the effort constraints, there exists an $(\hat S,\hat \alpha) \in \mathcal{C}$ such that
    \begin{align*}
        &\sum_{i \in \mathcal{N}} \sum_{j \in S_i} \hat U^P( \alpha_j,i)= \sum_{i \in \mathcal{N}} \sum_{j \in \hat S_i} \hat U^P(\hat \alpha_j,i)\\
        &\sum_{k \in  S_j} \hat U_{i,k}(\alpha_k) = \sum_{k \in \hat S_j} \hat U_{i,k}(\hat \alpha_k) \quad\forall i,j\in \mathcal{N}.
    \end{align*}
\end{proposition}

\begin{proof}
Property (i) follows from the feasibility check in Line \ref{alg2_line8ir}.
The second property follows from standard dynamic programming arguments that the algorithm keeps track of all possible utility profiles.
\end{proof}

The running time of the algorithm depends on the discretizations, which will be analyzed when instantiating it for $\eps$-EF contracts or EF1 contracts.
In particular, it depends on the number of possible utility tuples that can be induced.

\begin{algorithm}[t]
\caption{Dynamic programming to find all  utility profiles \\
\textbf{Input:} $\{r_j\}_{j \in \mathcal{M}}$, $\{p_{i,j}\}_{i \in \mathcal{N}, j \in \mathcal{M}}$, $\{c_{i,j}\}_{i \in \mathcal{N}, j \in \mathcal{M}}$, $m, n$, discretization of contracts $(\mathcal{D}_{j})_{j \in \mathcal{M}}$, of agents' utilities $(\mathcal{U}_{i})_{i \in \mathcal{N}}$, and principal's expected revenue $\mathcal{H}$, utility function $\hat U_{i,j}(\cdot)$ and $\hat U^P(\cdot,\cdot)$  \\
\textbf{Output:} set of  pairs (allocation, contracts)  $\mathcal{C}$.
}\label{algorithm_dynmaic_pr2}
$\mathcal{L}_0\gets \{(\emptyset^n,\emptyset, 0^{n^2+1})\}$\;
\For{$\ell = 1, \dots, m$}{
  $\mathcal{L}_{\ell} \gets \emptyset$ \;
  \For{$\hat{\alpha}_{\ell} \in \mathcal{D}_\ell$}{
    \For{\text{agent} $i = 1, \dots, n$}{
        \If{$\hat \alpha_\ell p_{i,\ell} r_\ell-c_{i,\ell} <0$ \label{alg2_line8ir}}{
        continue\;
        }   \label{alg2_line1oir}
      \For{  $(\hat{S}, \hat \alpha,y) \in \mathcal{L}_{\ell -1}$}{
           Let $\hat{S}_i' = \hat{S}_i \cup\{\ell\}$ \label{line:update1}\;
            Let $\hat{S}_j' = \hat{S}_j$ for each $j \neq i$ \;
           Let $\hat{\alpha}'=( \hat \alpha, \hat \alpha_\ell)$\;
        Compute $y'$ as per Equation \eqref{eq:update}\label{line:update4}\;
        \If{  there is no $(S, \alpha,y)\in \mathcal{L}_{\ell}$ such that $y=y'$ \label{line:remove}} {
        $\mathcal{L}_{\ell} = \mathcal{L}_{\ell} \cup  \{(\hat{S}', \hat \alpha',y')\}$ \;
        }
        }}}}
\Return $\mathcal{C}=\{\hat S,\hat \alpha: \textnormal{there exists a $y$ with } (\hat S,\hat \alpha,y) \in \mathcal{L}_m\}$.
\end{algorithm}

\paragraph{FPTAS Algorithms.}

We instantiate the dynamic-programming template to obtain FPTAS results that achieve almost optimality (regarding $\optval$) under $\eps$-EF and EF1 constraints, respectively. Our first result gives an FPTAS algorithm under $\eps$-EF constraints.

\begin{theorem}[FPTAS under $\eps$-EF Constraints]\label{thm:constant_agent}
When the number of agents is a constant, there is an algorithm that computes an $\eps$-envy-free contract with expected revenue at least $\optval-\eps$ in time polynomial in the instance size and $\frac{1}{\eps}$.
\end{theorem}

Next, we develop an FPTAS under the more challenging EF1 constraints.

\begin{theorem}[FPTAS under EF1 Constraints]\label{thm:constant_agentEF1}
When the number of agents is a constant, there is an algorithm that computes an EF1 contract with expected revenue at least $\optval-\eps$ in time polynomial in the instance size and $\frac{1}{\eps}$.
\end{theorem}

While the high-level frameworks to prove \cref{thm:constant_agent} and \cref{thm:constant_agentEF1} are similar, the proof of \cref{thm:constant_agentEF1} requires a delicate and non-trivial analysis. Below, we give a brief proof sketch for \cref{thm:constant_agentEF1}.

\vspace{2mm}

\noindent{\bf Key idea to prove \cref{thm:constant_agentEF1}.} Our algorithm framework starts by defining appropriate discretizations for contracts, agents' utilities, and the principal's expected revenue. Let $(S^*, \alpha^*)$ be the optimal EF allocation and contracts, and $\bar{\alpha}$ be the discretized contracts obtained from $\alpha^*$. Given this discrete environment, we apply \cref{thm:DP} to obtain a solution that gives the same discrete utility profile as $(S^*, \bar{\alpha})$. One trivial approach may be to uniformly discretize by some constant step. In this way, since the utilities of agents and the principal are additive, the obtained solution can finally lead to some $\eps$-approximate contracts in an additive sense. However, this approach fails to derive EF1 contracts, as it may require $\eps$ to be large, resulting in a large loss to the principal.
To circumvent this challenge, we develop a non-trivial adaptive discretization method to ensure that EF1 constraints can be recovered from additively violated EF constraints obtained above.

As the first step, we need to ensure that under the discretized optimal contract $(S^*, \bar\alpha)$,  the additive error introduced to agent $i$'s EF constraints has a dependence on its optimal utility $U_i^*$ (note that we can approximate it via polynomial enumeration). To achieve it, for the contract of each task $j$, we construct an upper bound (denoted as $\underline{\alpha}_j$) to be the minimum over {\it linear} contracts of all the agents, which gives the agent the maximum possible utility  $U_i^*$ from the single task $j$. For each agent $i$, we discretize once the contract space (denoted as $\mathcal{D}_{i,j}$) of task $j$ with a carefully designed step size $\delta(\underline{\alpha}_j -  \tau_{i,j})$. Hence, the discrete contract $\mathcal{D}_j$ of task $j$ is the union of all $\mathcal{D}_{i,j}$ for all the agent $i$.
By doing so, for every task $j$, we are able to ensure that agent $i$'s utilities differ at most $\delta U^*_i$ when exerting optimal efforts under $\alpha_j$ and $\bar{\alpha}_j$.
With such observation, we have the additive error in (approximately) EF constraints under $(S^*, \bar\alpha)$ as $\nu U^*_i$, where $\nu = m\delta$.

The second step is to translate the additive error to EF1 constraints.
To this end, we set the discretization step for agent $i$'s utility to be $\delta U_i^*$.
Recall that our dynamic programming finds a solution $(\hat{S}, \hat{\alpha})$ with the same {\it discrete} utility profile as  $({S}^*, \bar{\alpha})$.
Hence, under $(\hat{S}, \hat{\alpha})$, agent $i$'s utility before discretization, $\sum_{k\in \hat{S}_i} U_{i, k} (\hat{\alpha})$, differs from its discretized utility $\sum_{k \in \hat{S}_i} \hat{U}_{i, k} (\hat{\alpha})$ by at most $\nu U^*_i$ in an {\it additive} sense. Notably, the additive error can be converted into a {\it multiplicative} error, since $\sum_{k \in \hat{S}_i} \hat{U}_{i, k} (\hat{\alpha}) \ge U^*_i$ by definition.
Furthermore, as an intermediate step, we establish a multiplicatively approximate EF guarantee for any two agents $i, j$, $$\sum_{k \in \hat S_i} \hat \alpha_k p_{i, k }r_k - c_{i,k} \ge (1-6\nu) (\sum_{k \in \hat S_j} \max\{\hat \alpha_k p_{i, k} r_k - c_{i, k},0\}).$$
This follows from the previous observation, which converts the additive error depending on $U_i^*$---arising from the approximate EF constraints obtained in the first step and the guarantees of \cref{thm:DP}---into a multiplicative error.
Finally, to obtain an EF1 constraint, we remove from $\hat S_j$ the task that gives agent $i$ the largest utility. This induces a new factor of $\frac{1-6\nu}{1-1/m}$, which is greater than $1$ when $\nu$ is small enough.

\begin{proof}[Proof of \cref{thm:constant_agentEF1}]
    We will use the template of the $\epsilon$-EF dynamic programming approach designed in the previous section, enhanced with an initial guessing step and an adaptive grid.
    The pseudocode of our algorithm is in Algorithm \ref{algorithm_dynmaic_pr_EF1}.

    Denote the optimal solution of Program (\ref{pr:fair}) as $\alpha^*= (\alpha_1^*, \alpha_2^*, \dots, \alpha_m^*)$, and its corresponding fair allocations of tasks as $S^*= (S_1^*, S_2^*, \dots, S_n^*)$. Given the desired approximation $\epsilon>0$ for the principal's revenue,  if $m=1$, every full allocation is EF1 because the unique task can be removed from any compared nonempty bundle, so the problem reduces to choosing the revenue-maximizing effort-feasible assignment. Hence, we assume $m\ge 2$ and set errors $\nu= \min\{\epsilon, \frac{1}{6m}\}$ and $\delta=\frac{\nu}{m}$.

    \paragraph{Characterize the range of utilities.}

    Since, given an optimal allocation $S^*$, the optimal contract is a solution to an LP, standard arguments imply that the utility of each agent is $0$ or at least $2^{-f(I)}$, where $f(I)$ is a polynomial function of the instance size \citep{bertsimas1997introduction}.

    \paragraph{Guessing step.} We consider the following set of approximate agents' utilities:
    \[\mathcal{G}= \{0\} \cup \{ m 2^{-i}\}_{i \in \{0,1,\ldots, f(I)+\log(m)\}}. \]
    Then, we enumerate over all the tuples $(U_i)_{i \in \mathcal{N}}\in \mathcal{G}^{n}$ which include a guess of the utility for each agent.
     Notice that there exists a guess $(U_i)_{i \in \mathcal{N}}$ such that:
    \begin{align} \label{eq:guess}
    \sum_{k \in S^*_i} \alpha_k^* p_{i, k }r_k - c_{i,k} \in [U_i/2,U_i] \quad \forall i \in \mathcal{N}.
    \end{align}

    In the following, we focus on the execution of Algorithm \ref{algorithm_dynmaic_pr_EF1} when the guess is correct and hence we assume that Equation \eqref{eq:guess} holds.

    \paragraph{Adaptive grid.}
    Once we fix the utility profile $(U_i)_{i \in \mathcal{N}}$, we can define an adaptive grid for the dynamic programming as a function of the guesses.

    We define a set of possible discrete contracts as follows.
    Consider a task $j\in \mathcal{M}$ and an agent $i\in \mathcal{N}$. Let $\underline{\alpha}_{i,j}$ be the maximum $\alpha$ in the set
\begin{align}\label{def:barAlpha}
        \{\alpha\in [0,1]:\alpha p_{i,j} r_j-c_{i,j} \le U_i\}.
    \end{align}
    and let $\underline{\alpha}_j= \min_i \underline{\alpha}_{i,j}$.
    Then, we define the set
    \begin{equation}    \begin{aligned}\label{eq:disc_contracts}
        &\mathcal{D}_{i,j}= \begin{cases}
        \{ \tau_{i,j}+\delta k (\underline{\alpha}_j-\tau_{i,j})\}_{k \in \{0,1,\ldots,1/\delta\}} \quad \quad \textnormal{if }\underline{\alpha}_j\ge \tau_{i,j}\\
        \mathcal{D}_{i,j}= \emptyset \hspace{6.4cm} \textnormal{otherwise}
        \end{cases}
    \end{aligned}
    \end{equation}
    Notice that for each $j\in \mathcal{M}$, at least one of the sets $\{\mathcal{D}_{i,j}\}_{i \in \mathcal{N}}$ is not empty thanks to Assumption \ref{asp:usefulness}.
    Finally, we set the possible contracts for a task $j \in \mathcal{M}$ as
\begin{align}\label{eq:disc_contract_EF1}
        \mathcal{D}_j=\cup_{i\in \mathcal{N}} \mathcal{D}_{i,j}.
    \end{align}

    Now, we define the set of possible discretized utilities with a normalized uniform grid.
    In particular, for each agent $i$, we define the discretized utilities as
    \begin{align} \label{eq:disc_utilities_EF1}
        \mathcal{U}_i=\{\delta k U_i\}_{k \in \{0,1,\ldots,1/\delta\} } \quad\forall i \in \mathcal{N}
    \end{align}
    Finally, we use a uniform grid over the principal's expected revenue:
    \begin{align}\label{eq:discretization_principal_EF1}
        &\mathcal{H}=\{0,\delta, \ldots, 1\}
    \end{align}

Now, we are left to define the approximate utility functions $\hat U_{i,j}(\cdot)$, $\hat U^P(\cdot,\cdot)$. We let

\begin{equation}\label{eq:disc_utility_EF1}
\begin{aligned}
&\hat U_{i,j}(\hat \alpha_j)= \lceil \hat \alpha_j p_{i,j} r_j-c_{i,j}\rceil_{\mathcal{U}_i} \quad \forall  i \in \mathcal{N}, j \in \mathcal{M}, \hat \alpha_{j}\in \mathcal{D}_j\\
&\hat U^P(\hat \alpha_j,i)=\lceil (1-\hat \alpha_j) p_{i,j} r_j\rceil_{\mathcal{H}}\quad  \forall  i \in \mathcal{N}, j \in \mathcal{M}, \hat \alpha_{j}\in \mathcal{D}_j
\end{aligned}
\end{equation}
It is easy to see that $\lceil \hat \alpha_j p_{i,j} r_j-c_{i,j}\rceil_{\mathcal{U}_i}$ are well-defined thanks to the definition of $\mathcal{D}_j$.
Moreover, it holds
\begin{align}\label{eq:good_apx_EF1}
\hat U_{i,j}(\hat \alpha_j)\in [\max \{\hat \alpha_j p_{i,j} r_j-c_{i,j}, 0\},\max \{\hat \alpha_j p_{i,j} r_j-c_{i,j}, 0\} + \delta U_i]  \quad \forall  i \in \mathcal{N}, j \in \mathcal{M}, \hat \alpha_{j}\in \mathcal{D}_j
\end{align}

\paragraph{Existence of a good solution with discretized contracts.}
We start showing that there exists a discretized contract that is almost EF and almost optimal.
Consider the contract $\bar \alpha$ obtained by setting  $\bar \alpha_j = \lceil \alpha^*_j\rceil_{\mathcal{D}_j}$ for each $j \in \mathcal{M}$.
First of all, we have to show that such $\bar \alpha$ is well-defined. To do that, we have to guarantee that $\mathcal{D}_j$ includes at least  one contract greater than $\alpha^*_j$ for each $j \in \mathcal{M}$.

\begin{lemma}\label{wellefeindinelemmandk}
$\bar \alpha$ is well-defined, i.e., $\alpha^*_j \le \max_{\alpha_j \in \mathcal{D}_j} \alpha_j$ for all $j\in \mathcal{M}$.
\end{lemma}

Now, we show that $(S^*,\bar \alpha)$ is almost optimal and almost EF. Formally:

\begin{lemma}\label{lemma4_6fptasef1}
$(S^*,\bar \alpha)$ satisfies the effort constraints and guarantees:
\begin{equation}\label{discreteize_withinsec42}
\begin{split}
    \sum_{k \in S^*_i}\bar{\alpha}_k p_{i, k }r_k - c_{i,k}  \ge \sum_{k \in S^*_j} \max \{\bar\alpha_{k} p_{i,k}r_k -c_{i,k},0\}-\nu U_i \quad\forall i,j \in \mathcal{N},
\end{split}
\end{equation}
and
\begin{align}\label{eq:apx_opt_EF1}
\sum_{i \in \mathcal{N}} \sum_{k \in S^*_i} (1-\bar{\alpha}_k) p_{i,k}r_k \ge  \optval -\nu.
\end{align}
\end{lemma}

\paragraph{Find a good solution with discretized contracts.}

Our algorithm enumerates over all the possible utilities $(U_i)_{i \in \mathcal{N}}$. For each guess, the algorithm runs Algorithm \ref{algorithm_dynmaic_pr2}, and stores the set of EF1 contracts returned by it. Then, it outputs the contract  that maximizes the principal's expected revenue.
Let's denote this contract with $(\tilde{S}, \tilde{\alpha})$.

Now, by \cref{thm:DP}, Algorithm \ref{algorithm_dynmaic_pr_EF1} finds
a set $\mathcal{C}$ which includes an allocation $\hat S$ and a contract $\hat \alpha$ that induce the same discrete utility profile as $(\bar \alpha, S^*)$.
Formally, we have that
 \begin{align*}
        &\sum_{i \in \mathcal{N}} \sum_{j \in S_i^*} \hat U^P( \bar \alpha_j,i)= \sum_{i \in \mathcal{N}} \sum_{j \in \hat S_i} \hat U^P(\hat \alpha_j,i)\\
        &\sum_{k \in  S_j^*} \hat U_{i,k}(\bar\alpha_k) = \sum_{k \in \hat S_j} \hat U_{i,k}(\hat \alpha_k), \quad \forall i,j\in \mathcal{N}\\
        & \hat \alpha_j p_{i,j} r_j-c_{i,j} \ge 0, \quad \forall i \in \mathcal{N}, j \in \hat S_i.
\end{align*}
The final step is to show that $(\tilde S, \tilde \alpha)$ has principal's expected revenue at least $\optval -2\nu$, while by construction it is EF1 and satisfies the effort constraints.
To do so, it is sufficient to show that $(\hat S, \hat \alpha)$ is EF1 and $2\nu$-optimal, since such a contract belongs to $\mathcal{C}$.

\begin{lemma}\label{lemma4_7_epftsef1}
 $(\hat S, \hat \alpha)$  guarantees
$\sum_{i \in \mathcal{N}} \sum_{k \in \hat S_i} (1-\hat{\alpha}_k) p_{i,k}r_k \ge \optval-2\nu$.
\end{lemma}

Hence, we are left to show that the contract $(\hat S, \hat \alpha)$  is indeed EF1.
\begin{lemma}
 $(\hat S, \hat \alpha)$ satisfies \text{\rm EF1}.
\end{lemma}
\begin{proof}
Take two agents $i$ and $j$.
Then
\begin{subequations}\label{eq:bound_multi}
\begin{align}
     \sum_{k \in \hat S_i} \hat U_{i,k}(\hat \alpha) &= \sum_{k \in  S^*_i}  \hat U_{i,k}(\bar \alpha)
     \overset{ \text{by definition of $\hat{U}_{i,k}$}
}{\ge} \sum_{k \in S^*_i}\bar{\alpha}_k p_{i, k }r_k - c_{i,k} \label{eq:bound_multi_2}\\
    & \overset{\text{by (\ref{discreteize_withinsec42})}}{\ge} \sum_{k \in S^*_j} \max\{ \bar{\alpha}_k p_{i, k }r_k - c_{i,k}, 0\} -\nu U_i
     \overset{\text{by \eqref{eq:good_apx_EF1}}}{\ge} \sum_{k \in S^*_j} \hat U_{i,k}(\bar \alpha) -2\nu U_i \label{eq:bound_multi_4}\\
    &  = \sum_{k \in \hat S_j} \hat U_{i,k}(\hat \alpha) -2\nu U_i
     \ge \sum_{k \in \hat S_j} \max\{ \hat \alpha_k p_{i,k} r_k - c_{i,k}, 0\}- 2\nu U_i ,
\end{align}
\end{subequations}

To conclude the proof, we want to show that there exists a task $k_{i,j}$ such that
    \[ \sum_{k \in \hat S_i} \hat \alpha_k p_{i, k }r_k - c_{i,k} \geq \sum_{k \in \hat S_j\backslash\{k_{i,j}\}} \max\{\hat \alpha_k p_{i, k} r_k - c_{i, k},0\}.\]

    Take the task giving the largest utility $k_{i,j}\in \arg \max_{k \in \hat S_{j}} \hat \alpha_k p_{i,k} r_k -c_{i,k}$.
   Then, we have that
\begin{align*}
         \sum_{k \in \hat S_i} \hat \alpha_k p_{i, k }r_k - c_{i,k} & \geq \sum_{k \in \hat S_i} \hat U_{i,k}(\hat \alpha_k) - \nu U_i
          = (\sum_{k \in \hat S_i} \hat U_{i,k}(\hat \alpha_k) +2 \nu U_i) -3 \nu U_i\\
         & \geq (1-6\nu) (\sum_{k \in \hat S_i} \hat U_{i,k}(\hat \alpha_k)+ 2 \nu U_i)
           \geq (1-6\nu) (\sum_{k \in \hat S_j} \max\{\hat \alpha_k p_{i, k} r_k - c_{i, k},0\})\\
         & \geq \frac{1-6\nu}{1-1/m} \sum_{k \in \hat S_j\setminus k_{i,j}} \max\{\hat \alpha_k p_{i, k} r_k - c_{i, k},0\}
         \geq \sum_{k \in \hat S_j\setminus k_{i,j}} \max\{\hat \alpha_k p_{i, k} r_k - c_{i, k},0\}
\end{align*}
where the first inequality is by Equation \eqref{eq:good_apx_EF1}, the second inequality is by Equation \eqref{eq:guess} and that $\sum_{k \in \hat S_i} \hat U_{i,k}(\hat \alpha_k) = \sum_{k \in  S_i^*} \hat U_{i,k}(\bar\alpha_k) \ge \sum_{k \in S^*_i} \alpha_k^* p_{i, k }r_k - c_{i,k} \ge \frac{U_i}{2}$,
the third inequality is by Equation \eqref{eq:bound_multi}, the fourth inequality is by the definition of $k_{i,j}$, and the last inequality is by $\nu\le \frac{1}{6m}$.
\end{proof}

We conclude the proof by noticing that the algorithm runs in polynomial time. Indeed, Algorithm \ref{algorithm_dynmaic_pr_EF1} clearly runs in polynomial time excluding the calls to Algorithm \ref{algorithm_dynmaic_pr2}. Moreover, it is easy to see that each call to Algorithm \ref{algorithm_dynmaic_pr2} runs in polynomial time since the possible vectors of agents' cumulative utilities and the principal's cumulative expected revenue have polynomial size (agent $i$ utility belongs to $\{0,\delta U_i, \ldots,m U_i\}$, while the principal's expected revenue belongs to $\{0,\delta, \ldots,m\}$ ).
\end{proof}

\begin{algorithm}[t]
\caption{Computing EF1 Contracts \\
\textbf{Input:} $\epsilon>0$ $\{r_j\}_{j \in \mathcal{M}}$, $\{p_{i,j}\}_{i \in \mathcal{N}, j \in \mathcal{M}}$, $\{c_{i,j}\}_{i \in \mathcal{N}, j \in \mathcal{M}}$, $m, n$ \\
\textbf{Output:} The contracts $(\tilde{\alpha}, \tilde{S})$.
}\label{algorithm_dynmaic_pr_EF1}
$\nu \gets \min\{\epsilon, \frac{1}{6m}\}$\;
$\delta\gets \nu/m$\;
$\mathcal{C}^* \gets \emptyset$ \;
\For {$(U_i)_{i \in \mathcal{N}}\in \mathcal{G}^{n}$}{
define  $\mathcal{D}_j$, $\mathcal{U}_i$, $\mathcal{H}$, $\hat U_{i,j}(\cdot)$, $\hat U^P(\cdot,\cdot)$ as per Equation \eqref{eq:disc_contract_EF1}, \eqref{eq:disc_utilities_EF1}, \eqref{eq:discretization_principal_EF1}, and \eqref{eq:disc_utility_EF1}\;
$\mathcal{C}\gets$ Run Algorithm \ref{algorithm_dynmaic_pr2} with discretization  $\mathcal{D}_j$, $\mathcal{U}_i$, $\mathcal{H}$, and functions $\hat U_{i,j}(\cdot)$, $\hat U^P(\cdot,\cdot)$\;
\For{ $( S,\alpha) \in \mathcal{C}$}{ \label{line_5}
   \If{$(S, \alpha)$ satisfies EF1 }{
$\mathcal{C}^* = \mathcal{C}^* \cup \{{S},\alpha\}$ \;
   }
} \label{line9}
}
\Return $({\tilde\alpha}, {\tilde S}) \in \mathcal{C}^*$ with maximum  principal's expected revenue.
\end{algorithm}

\section{Price of Fairness}\label{sec:priceoffairness}

In this section, we examine how fairness constraints affect the principal's revenue.
Let $\opt$ be the optimal revenue the principal can guarantee without fairness concerns.
Ideally, the principal would hope to keep a large fraction of $\opt$ using envy-free contracts.
To quantify the loss of revenue due to fairness constraints, we introduce the notion of \emph{price of fairness} (PoF), which is defined as the multiplicative gap between $\opt$ and the optimal revenue under various fairness constraints, such as $\opt/\optval$ for EF contracts and $\opt/\optefw$ for EF1 contracts.
We similarly use $\opt/\opteps$ for $\eps$-EF contracts. 

\subsection{EF: Unbounded Price of Fairness}

We first focus on exact envy-free contracts and show that the price of fairness is unbounded even in the case of only two agents. 
\begin{proposition}\label{price_fair}
For any constant $c > 1$, there exists an instance with only two agents where the
optimal revenue from envy-free contracts is less than a $\frac{1}{c}$ fraction of the optimal revenue without fairness concerns, i.e., $\optval < \frac{1}{c}\cdot \opt$.
\end{proposition}
This result can be proved with the following example.
\begin{example}\label{example52}
Consider an instance parameterized by a constant  $\delta$$>0$. Suppose there is one task with reward $r=1$ and two agents.  Set the success probabilities and costs to $p_{1} = 10 \delta$, $c_1 =  \delta$ and $p_{2} = \frac{1}{2}$, $c_2 = \frac{1}{4}$. Note that in this case, linear contracts are without loss of optimality: If the principal pays agents on the failure outcome, the agent assigned no task would envy, as it would at least obtain positive utility by choosing to shirk if assigned a task.

To achieve envy-freeness, the principal must assign the task to agent $1$, and to maximize the revenue, the principal needs to offer a linear contract $\frac{c_1}{p_1}=\frac{1}{10}$ and gains a revenue of $(1-\frac{1}{10}) \cdot 10 \delta = 9 \delta$.
However, without the EF constraints, the principal would assign tasks to agent $2$ and gain a revenue of $\frac{1}{4}$ by offering a linear contract $\frac{1}{2}$.
The price-of-fairness ratio $\opt / \optval $ is $\frac{1}{36 \delta}$, which becomes arbitrarily large when  $\delta$$\to 0$.
\end{example}

Recall that the unbounded gap in \cref{price_fair} holds even when arbitrary contracts are allowed. This suggests that the high price of fairness is driven primarily by the imposed fairness constraints, rather than by restrictions on the class of linear contracts.
Therefore, in the next two subsections, we consider two ways to reduce the PoF: (i) we relax exact envy-freeness to $\eps$-EF, which permits a small additive amount of envy; and (ii) we consider EF1 constraints, which introduce unfairness up to one task.

\subsection{\texorpdfstring{$\eps$}{epsilon}-EF: Small Additive Relaxations}

Recall that rewards are normalized so that $r_j\le 1$ for every task $j$, and hence each task-level expected value $p_{i,j}r_j$ is also at most $1$. Thus, the additive slack $\eps$ in the $\eps$-EF constraints is measured on the same normalized utility scale as the agents' task-contract values.

The next result shows that, in the small-$\eps$ regime, additive envy-freeness has price of fairness of order $1/\eps$ for fixed $n$ and $m$.

\begin{proposition}\label{prop:pof_eps_ef}
For every $n\ge 2$, $m\ge 1$, and $0<\eps\le 1/4$, the optimal revenue from $\eps$-EF contracts is at least a $\frac{4\eps}{\min\{m,n^2\}}$ fraction of the optimal revenue without fairness concerns, i.e.,
\[
\opteps
\ge
\frac{4\eps}{\min\{m,n^2\}}\cdot \opt.
\]
Moreover, for every $0<\eps<1/4$ and every $\rho>0$, there exists an instance with $n$ agents and $m$ tasks for which
\[
\opteps
\le
(4\eps+\rho)\cdot \opt.
\]
\end{proposition}

The proof is deferred to \cref{apx:eps_ef_pof}. The positive guarantee follows from the better of two constructions: one spends an envy budget of $\eps/m$ independently on every task, while the other uses a round-robin construction over carefully chosen single-task contracts. The lower-bound example is a single-task instance. Thus, when $m=1$, the two sides coincide up to an arbitrarily small $\rho$; equivalently, the price-of-fairness bound is tight for a single task, and equals $1/(4\eps)$ for every $0<\eps\le 1/4$. More generally, for every fixed $n$ and $m$, the proposition gives an asymptotically tight dependence on $\eps$: the worst-case multiplicative gap $\opt/\opteps$ is $\Theta_{n,m}(1/\eps)$ as $\eps\downarrow 0$.

\subsection{EF1: Introducing Unfairness up to One Task}

In this section, we consider EF1 constraints and show that at least a $\frac{1}{n^2}$ fraction of the revenue can be extracted by EF1 contracts. Note that this provides a constant approximation ratio if $n$ is a constant.
This retained fraction is meaningful if the number of agents is not too large. For instance, the EF1 revenue in \cref{example52} sharply increases to $\frac{1}{4}$.  Furthermore, we complement the previous result by showing that the retained fraction can be as small as $O(1/\sqrt{n})$, or equivalently that the price of fairness can be as large as $\Omega(\sqrt{n})$. The construction of the instance takes inspiration from the construction in \cite{bei2021price}, but requires a delicate adaptation and more involved analysis. We leave as an open problem to close this gap.

\begin{proposition}\label{ef1upperboundpro}
For any number of agents $n$, there exists an EF1 contract with expected revenue at least a $\frac{1}{n^2}$ fraction of the optimal revenue without fairness concerns, i.e., 
\[\optefw\geq \frac{1}{n^2}\cdot \opt.\]
Moreover, for any number of agents $n\ge 9$, there exists an instance where the optimal revenue from EF1 contracts is at most an $O(\frac{1}{\sqrt{n}})$ fraction of the optimal revenue without fairness concerns, i.e., \[
\optefw \leq O(\frac{1}{\sqrt{n}}) \cdot \opt.\]
\end{proposition}

\begin{proof}
In this proof, we prove the upper bound and lower bound for the PoF separately.

\paragraph{The upper bound $O(n^2)$.}
Let $i^*$ be the agent that maximizes $\sum_{j \in \mathcal{M}}\max\{p_{i,j}r_j-c_{i,j}, 0\}$. Note that this maximum value is at least $\frac{\opt}{n}$, where we recall that $\opt$ is the optimal principal's expected revenue without fairness concerns.
Set the contract for each task as $\alpha_j=\tau_{i^*,j}$ if $p_{i^*,j}r_j-c_{i^*,j} \geq 0$; otherwise, set $\alpha_j = \min_{i \in \mathcal{P}_j} \tau_{i,j}$ where $\mathcal{P}_j=\{i \in \mathcal{N}\mid p_{i,j}r_j-c_{i,j} \geq 0\}$. We recall that this set is not empty thanks to Assumption \ref{asp:usefulness}.

Now, we construct the contract using a round-robin algorithm, where agents take turns to select one of the remaining tasks. In particular, we start with agent $i^*$, followed by the other agents in an arbitrary order. Let $\hat S$ be the set of remaining tasks. In each turn, each agent $i$ chooses one of the remaining tasks in $\arg\max_{j \in \hat S} \alpha_{j} p_{i,j} r_{j}-c_{i,j}$ with the following additional constraints:
\begin{itemize}
    \item the agent selects the tasks by breaking ties to maximize the principal's revenue;
    \item the agent does not select a task if no remaining task satisfies her effort constraint.
\end{itemize}
It is straightforward to see that the modifications do not affect the conclusion that the allocation given by the round-robin algorithm is EF1.
Moreover, contracts are set in order to guarantee that for any task there exists an agent that can be incentivized to exert effort. Hence, our modified round-robin algorithm can ensure that the resulting allocation is a full allocation.

Now, we provide a lower bound on the principal's expected revenue. To do that, it is sufficient to focus on agent $i^*$.
Since agent $i^*$ has zero utility for all tasks by construction, and in each round, the tie-breaking rule assigns to $i^*$ the task that provides the principal's maximum revenue (among the remaining ones),
it is easy to see that the principal has revenue at least
\[\frac{1}{n} \sum_{j \in \mathcal{M}}\max\{p_{i^*,j}r_j-c_{i^*,j}, 0\} \ge \frac{1}{n^2}\opt,\]
implying
$\optefw  \geq \frac{1}{n^2}\cdot \opt$.

\paragraph{The lower bound of PoF $\Omega(\sqrt{n})$.} 

Let $m=n$ with $n \ge 9$. Also, let $r=\floor{\sqrt{n}} \ge 3$. The rewards are $1$ for all the tasks. Next, we define the success probabilities and costs for each agent.
\begin{itemize}
    \item For agent $i \in \{1, 2, \dots, r-1\}$, let the probabilities and costs be $p_{i, (i-1)r +k } = 1$, $c_{i, (i-1)r +k } = \frac{r-1}{r}$ for $k =  1, 2, \dots, r$, while for all other tasks $j$, let $p_{i, j } = 0$, $c_{i, j } = 1$.
    \item For agent $i=r$, let $p_{i, j} = 1$, $c_{i, j } = \frac{n-r(r-1)-1}{n-r(r-1)}$ for tasks $j \in \{ r(r-1)+1, r(r-1)+2,  \dots, n\}$, while for all other tasks $j$, let $p_{i, j} = 0$, $c_{i, j} = 1$.
    \item For agent $i\in \{r+1, \dots, n\}$,  let $p_{i, j} = \frac{2}{n}$, $c_{i, j } = \frac{1}{n}$ for all the task $j \in \mathcal{M}$.
\end{itemize}

Note that assigning one task to any agent $i \in \{ r+1,\dots, n\}$ yields revenue at most $\frac{1}{n}$ to the principal.
To obtain the maximum revenue without fairness, the principal assigns
to each agent $i\in \{1, 2, \dots, r-1\}$ tasks $j \in \{(i-1)r+1, \dots, ir\}$ and sets the contract $\alpha_j = \frac{r-1}{r}>\frac{1}{n}$. The principal's revenue from each of these contracts is $\frac{1}{r}$.
Moreover, the principal assigns task $j\in \{r(r-1)+1, \dots, n\}$  to agent $r$ with contract $\alpha_j = \frac{n-r(r-1)-1}{n-r(r-1)}$. The revenue from each of these tasks is $\frac{1}{n-r(r-1)} > \frac{1}{n}$.
This leads to a cumulative principal's revenue of $(r-1) \cdot r\cdot \frac{1}{r} + (n-r(r-1))\cdot \frac{1}{(n-r(r-1))} = r$.

Next, we provide an upper bound for the revenue achievable with EF1 contracts.
Note that by construction, agent $i \in \{1, 2, \dots, r-1\}$ can only be assigned tasks $j \in \{(i-1)r+1, \dots, ir \}$, and agent $i=r$ can only be assigned tasks $j\in \{r(r-1)+1, \dots, n\}$.

 We claim that if there exists one agent $i \in \{1, 2, \dots, r\}$ assigned at least two tasks, while there is one agent $i' \in \{r+1, \dots, n\}$ assigned no tasks, then agent $i'$  will envy agent $i$. Note that by construction, the minimum contract to incentivize agent $i'$ is $\frac{1}{2}$. Then, the claim holds because the minimum contract that incentivizes agent $i \in \{1, 2, \dots, r-1\}$ and agent $i=r$ is greater than that for agent $i \in \{r+1,r+2, \dots, n\}$, i.e.,  $\frac{r-1}{r} > \frac{1}{2}$ and $\frac{n-r(r-1)-1}{n-r(r-1)}\ge \frac{r-1}{r} > \frac{1}{2}$ by $n > 9$.

Recall that there are $n$ tasks. As long as there are at least two tasks assigned to one agent in $\{1, 2, \dots, r\}$, each agent in $\{r+1, \dots, n\}$ must be assigned at least one task. Note that the maximum revenue from assigning task $j\in \{r(r-1)+1, \dots, n\}$ to agent $i=r$ is $\frac{1}{n-r(r-1)} < \frac{1}{r}$ since $r = \floor{\sqrt{n}}$ implies $n-r(r-1) \ge r$.
Hence, the  revenue of EF1 contracts in this case is at most $(n-r)\cdot \frac{1}{n} +  (\frac{1}{r})\cdot(r)<2$, which is exactly the  welfare extracted by assigning  one task to each agent in $\{r+1, \dots, n\}$ while the remaining tasks are assigned to agents in $\{1, 2, \dots, r\}$.

Moreover, if all the agents in $\{1, 2, \dots, r\}$ are assigned at most one task, then the revenue of EF1 contracts is also at most the welfare $(n-r)\cdot \frac{1}{n} +  (\frac{1}{r})\cdot(r)<2$. This is because assigning more tasks to agents in $\{r+1, \dots, n\}$ will lead to lower welfare, as $\frac{1}{r}\ge\frac{1}{n-r(r-1)}>\frac{1}{n}$, i.e., the maximum welfare extracted from an agent in $\{r+1, \dots, n\}$ completing one task is less than that from an agent in $\{1, 2, \dots, r\}$.

In summary, the optimal revenue from EF1 contracts is at most $2$, which leads to a retained fraction of at most $\frac{2}{r} = \frac{2}{\floor{\sqrt{n}}}$. This concludes the proof.
\end{proof}

\section{Conclusions}
\label{sec:conclude}

This paper initiates the study of a new class of fair contract design problems that are both practically relevant and theoretically rich, providing a unified framework to understand how fairness constraints affect a principal’s profit maximization in settings such as algorithmic labor platforms and task allocation systems. By integrating ideas from contract theory, fair allocation, and combinatorial optimization, our results open a new research direction at the intersection of economics and computation.

Several important questions remain open in our paper. First, it is unclear whether one can obtain any non-trivial constant approximation ratios to the optimal EF benchmark under EF1 constraints, or under exact EF constraints when restricting attention to the case with a constant number of agents.
Second, for the price of fairness, our upper and lower bounds leave a gap for both $\eps$-EF and EF1 contracts that depends on the number of agents, and closing it remains an interesting direction for future work.

Beyond our specific model, we identify three promising directions for future research that may be of broader interest to the community:
\begin{itemize}
    \item While we focus on linear contracts for tractability and interpretability, an interesting direction is to study more expressive (e.g., nonlinear) contract classes under fairness constraints and understand their effects on the complexity and structure of fair contract design.
    \item Beyond profit maximization, welfare maximization is an important objective, especially in public-sector settings such as government project allocation. In this context, fairness remains central, but contract design must still address incentive compatibility, distinguishing it from prior work on fair allocation without strategic agents~\citep[e.g.,][]{bu2022complexity,aziz2023computing}.
    \item Since our hardness results show that even approximate optimization under envy-freeness is \NP-hard, it is natural to explore alternative fairness notions—such as proportionality or maximin share guarantees~\citep[e.g.,][]{suksompong2016asymptotic,garg2020improved}—that may offer a better trade-off between fairness and computational tractability.
\end{itemize}

\newpage
\bibliographystyle{apalike}
\bibliography{ref}

\newpage
\appendix

\section{Equal-Pay Contracts Need Not Be Envy-Free}
\label{apx:equal_pay_not_fair}

This section gives a simple example showing that an equal-pay restriction can be incompatible with envy-freeness. The example uses identical agents and identical success probabilities, so the incompatibility arises already from cost heterogeneity across tasks. Consider two agents and two tasks, with rewards $r_1=r_2=1$. For both agents $i\in\{1,2\}$, let
\[
    p_{i,1}=p_{i,2}=1,\qquad c_{i,1}=\frac{1}{4},\qquad c_{i,2}=\frac{1}{2}.
\]
Thus task $1$ is strictly more attractive than task $2$ for every agent at any common payment level because it has a lower cost.

Suppose the principal imposes equal pay across the two tasks, interpreted in this task-level model as equal success-contingent payments: $\alpha_1r_1=\alpha_2r_2$. Since the rewards are equal, this is equivalent to $\alpha_1=\alpha_2=\alpha$; in particular, the smallest equal-pay contract that can assign both tasks while satisfying the effort constraints is $\alpha=\frac{1}{2}$. Under this contract, each agent's utility from task $1$ is
\[
    \alpha p_{i,1}r_1-c_{i,1}=\frac{1}{2}-\frac{1}{4}=\frac{1}{4},
\]
whereas each agent's utility from task $2$ is
\[
    \alpha p_{i,2}r_2-c_{i,2}=\frac{1}{2}-\frac{1}{2}=0.
\]
No full allocation can be envy-free. If the two tasks are assigned to different agents, then the agent receiving task $2$ envies the agent receiving task $1$. If both tasks are assigned to the same agent, then the other agent obtains zero utility but values the assigned bundle at $\frac{1}{4}$, and hence envies the assigned agent.

The same argument rules out every equal-pay full-allocation contract in this instance. If $\alpha<\frac{1}{2}$, task $2$ cannot be assigned to any agent while satisfying the effort constraint. If $\alpha\ge \frac{1}{2}$, every agent strictly prefers task $1$ to task $2$, so the case analysis above applies. By contrast, without equal pay, the principal can set $\alpha_1=\frac{1}{4}$ and $\alpha_2=\frac{1}{2}$ and assign one task to each agent; all induced utilities are zero, so the resulting full-allocation contract is envy-free.

\section{Missing Proofs and Results in \cref{section_model_part}}

\subsection{Proof of \cref{prop:ef_exists}}

\begin{proof}
    One possible EF contract is constructed as follows: for every task $j$, among the set of agents $\mathcal{M}_j = \{i \mid p_{i,j}r_j -c_{i,j} \ge 0\}$, assign task $j$ to the agent $i^* \in \mathcal{M}_j$ where $i^* =\argmin_i \tau_{i,j}$ and set the contract $\alpha_j = \tau_{i^*,j}$. The resulting allocation is EF and satisfies the effort constraints.
\end{proof}

\subsection{Proof of \cref{prop_full_allo}}
\begin{proof}
    Suppose that in an envy-free contract $(S,\alpha)$, task $j$ is unallocated. Among the set of agents $\mathcal{M}_j = \{i \mid p_{i,j}r_j -c_{i,j} \ge 0\}$, assign task $j$ to the agent $i^* \in \mathcal{M}_j$ where $i^* =\argmin_i \tau_{i,j}$ and set the contract $\alpha_j = \tau_{i^*,j}$. Apply the same operation to every other unallocated task. These modifications preserve the effort and EF constraints while increasing the principal's revenue.
\end{proof}

\section{Missing Proofs and Results in \cref{sec:general}} \label{missingproof_sec_gener}

\subsection{Proof of \cref{hardness_ef}}

\begin{proof}
    Our proof is by a reduction from the partition problem:
    given a set of $\ell$ integers $n_1, n_2, \dots, n_{\ell}$, find a subset of integers $\mathcal{S}$ such that $\sum_{j\in \mathcal{S}} n_j = \sum_{j\in [\ell]} n_j/2$.

    \paragraph{Construction.}
    We construct an instance where there are $3$ agents, i.e., $n=3$.
    Given a set of integers $n_1, n_2, \dots, n_{\ell}$, we define the instance as follows. There exists a task $j^*$ whose reward is  $r_{j^*}=1$. The agents' success probabilities for task $j^*$ are  $p_{1,j^*}=1$,  $p_{2,j^*}=p_{3,j^*}=\frac{1}{10}$, and their costs  are $c_{1,j^*}=\frac{1}{2}$, $c_{2,j^*}=c_{3,j^*}=0$.
        Let $M=10 \sum_{j \in [\ell]}n_j$.
    For each $n_j$, we define a task $j$ with reward $r_j  = n_j/M$. Additionally, we define the agents' success probabilities as  $p_{1,j}=0$, $p_{2,j}=p_{3,j}=1$ and costs as $c_{1,j}=c_{2,j}=c_{3,j}=0$.

\paragraph{Sufficiency.}
    We show that if there exists a partition, then the principal's revenue is at least $\frac{1}{2}$.

Suppose that there is a set of indexes $\mathcal{S}$ such that $\sum_{j\in \mathcal{S}} n_j=\sum_{j \in [\ell]} n_j/2$.
Consider the following contract. Assign task $j^*$ to agent $1$ with contracts as $\alpha_{j^*}=\frac{1}{2}$.
Moreover, assign all the tasks in $\mathcal{S}$ to agent $2$, and  all the other tasks to agent $3$.
We set the contracts
$\alpha_j=1$ to all the tasks $j\in [\ell]$.

We are left to show that the allocation satisfies the effort constraints and guarantees revenue at least $\frac{1}{2}$.
First, the allocation is envy-free. Indeed, by the construction above, agent $1$ receives utility $0$. Due to $p_{1,j}=0$ for all $j\in [\ell]$, agent $1$ receives utility $0$ switching to the tasks assigned to one of the other agents.
Agent $2$ receives utility
\[\sum_{j\in \mathcal{S}} r_j p_{2,j} - c_{2,j} =  \sum_{j\in \mathcal{S}} \frac{n_j}{M}=\frac{1}{20}.\]
Switching to the tasks assigned to agent $1$, agent $2$ gets utility $\alpha_{j^*}r_{j^*}p_{2, j^*} = \frac{1}{20}$. In addition, switching to the task assigned to agent $3$, agent $2$ gets utility $\sum_{j \in [\ell]\setminus \mathcal{S}} \alpha_{j}r_{j}p_{2, j} - c_{2, j} = \frac{1}{20}$ since $\mathcal{S}$ is a partition.
Hence, the agent $2$ does not envy other agents.
A similar argument applies to agent $3$.

We conclude noticing that the principal gets positive revenue only from task $j^*$ and in particular $\frac{1}{2}$.

\paragraph{Necessity.} We prove that if there is no partition, the principal's revenue is at most $\frac{1}{5}$.

    We prove the result by contradiction. Assume that the principal's revenue is strictly greater than $\frac{1}{5}$. Notice that if agent $1$ is not assigned the task  $j^*$, the principal's revenue is at most:
    \[
    r_{j^*}p_{2,j^*} + \sum_{j\in [\ell]} r_{j}p_{2,j} = \frac{1}{10} + \sum_{j\in [\ell]} \frac{n_j}{M} = \frac{1}{5},
    \]
    which is obtained, for instance, by assigning all the tasks to agent $2$.

    Hence, task $j^*$ must be assigned to agent $1$ since by assumption the principal's revenue is strictly greater than $\frac{1}{5}$.
    This directly implies that the payment for task $j^*$ is at least $\alpha_{j^*}\ge \frac{1}{2}$ by the effort constraint of agent $1$.

    Now, we analyze the EF constraints of agent $2$ and $3$.
    Since agent $2$ does not envy agent $1$, the set of tasks $S_2\subseteq[\ell]$ assigned to agent $2$ must guarantee utility:
    \[\sum_{j \in S_2} \alpha_j r_j p_{2,j} =  \sum_{j \in S_2} \alpha_j \frac{n_j}{M}\ge \alpha_{j^*}r_{j^*}p_{2, j^*} \ge 1/20.\]
    where the first inequality is by the EF constraints with respect to agent $1$, and the second inequality is by $\alpha_{j^*}\ge \frac{1}{2}$.
Therefore, we have that
    \[\sum_{j \in S_2} \frac{n_j}{10 \sum_{j\in [\ell]}n_j} =\sum_{j \in S_2} \frac{n_j}{M} \ge \sum_{j \in S_2} \alpha_{j}\frac{n_j}{M}\ge \frac{1}{20}\]
    which further implies that
     \[\sum_{j \in S_2} n_j\ge \frac{1}{2} \sum_{j\in[\ell]} n_j.\]
     By similar arguments, we have that for the set of tasks $S_3\subseteq [\ell]$ assigned to agent $3$ guarantees:
      \[\sum_{j \in S_3} n_j\ge \frac{1}{2} \sum_{j\in [\ell]} n_j.\]

Hence, it must be the case that $\sum_{j \in S_2} n_j = \sum_{j \in S_3} n_j= \frac{1}{2} \sum_{j\in [\ell]} n_j$, which implies that  $S_2$ defines a partition, leading to a contradiction.

We conclude the proof by noting that any algorithm guaranteeing a revenue fraction strictly larger than $\frac{1}{5}/\frac{1}{2}=\frac{2}{5}$ would distinguish the two cases; equivalently, no approximation ratio strictly smaller than $\frac{5}{2}$ is possible.
\end{proof}

\subsection{Proof of \cref{prop:two_ef_hard}}
\label{subapx:hardness}

\begin{proof}[Proof of \cref{prop:two_ef_hard}]
We prove the results by reducing from the partition problem:  given a set of $\ell$ integers $n_1, n_2, \dots, n_{\ell}$, find a subset of integers $\mathcal{S}$ such that $\sum_{j\in \mathcal{S}} n_j = \sum_{j\in [\ell]} n_j/2$.

    \paragraph{Construction.}
    There are two agents $n=2$ and $m=\ell+2$ tasks.
    There is a task $j^*$ with reward $r_{j^*}=1$. The success probabilities for task $j^*$ are  $p_{1,j^*}=1$,  $p_{2,j^*}=\frac{1}{10}$, and the related costs are $c_{1,j^*}=\frac{1}{2}$, $c_{2,j^*}=0$. Moreover, there is one task $\bar{j}$ with reward $r_{\bar{j}} =1$. The success probabilities relative to task $\bar j$ are  $p_{1,\bar{j}} = \frac{1}{10}, p_{2,\bar{j}} = 0$, and the costs are $c_{1,\bar{j}} = c_{2,\bar{j}} = 0$.

    Additionally, let $T=\sum_{j\in S} n_j$ and $M=5 T$. Then, for each $n_j\in \mathcal{S}$, we add a task $j$ with reward $r_j = \frac{n_j}{M}$. The success probabilities relative to this task are  $p_{1,j}=1$, $p_{2,j}=\frac{1}{2}$, and the costs are $c_{1,j}=c_{2,j}=0$.

\paragraph{Sufficiency.} Suppose there exists a solution $\mathcal{S}_1,\mathcal{S}_2$ to the partition problem, i.e., $\mathcal{S}_1\cup \mathcal{S}_2 = \mathcal{S}$ and $\sum_{j \in S_1} n_j = \sum_{j \in S_2} n_j = \frac{T}{2}$. We can construct the following contracts:
\begin{itemize}
    \item $\alpha_{j^*} = \frac{1}{2}$
    \item $\alpha_j = 0$ for $j\in \mathcal{S}_1$
    \item $\alpha_j = 1$ for $j \in  \mathcal{S}_2$
    \item $\alpha_{\bar{j}}=1$
\end{itemize}
The tasks assigned to agent $1$ are $G_1 = \mathcal{S}_1\cup \{j^*, \bar{j}\}$, and the tasks assigned to agent $2$ are $G_2 = \mathcal{S}_2$. First, by construction, the contracts satisfy the effort constraints. Next, we verify that the allocations are EF. For agent $1$, the EF constraints are satisfied, i.e.,
\begin{align*}
    \sum_{j\in S_1} r_j \alpha_j p_{1,j} + (r_{j^*} \alpha_{j^*}p_{1, j^*} - c_{1,j^*}) + r_{\bar{j}} \alpha_{\bar{j}} p_{1, \bar{j}}&= 0 + \frac{1}{10} \ge \sum_{j\in S_2} r_j \alpha_j p_{1,j}=  \sum_{j\in S_2} \frac{n_j}{M} = \frac{\frac{1}{2}T}{5T} =  \frac{1}{10}
\end{align*}
The EF constraints for agent $2$ are also satisfied,
\begin{align*}
&\sum_{j\in S_2} r_j \alpha_j p_{2,j}  = \frac{1}{2} \sum_{j\in S_2} \frac{n_j}{M}=\frac{1}{20} \ge \sum_{j\in S_1} r_j \alpha_j p_{2,j} + (r_{j^*} \alpha_{j^*}p_{2, j^*} - c_{2,j^*}) + r_{\bar{j}} \alpha_{\bar{j}} p_{2, \bar{j}} = \frac{1}{20}
\end{align*}
Hence, we have that the principal's revenue is
\[ (1-\alpha_{j^*})r_{j^*}p_{1,j^*} + \sum_{j\in S_1} (1-\alpha_j) r_j p_{1, j}= \frac{1}{2} + \sum_{j\in S_1} \frac{n_j}{M} = \frac{1}{2} + \frac{1}{10}. \]

\paragraph{Necessity.}
Next, we show that if the partition problem is not satisfiable, then the optimal revenue is strictly smaller than $\frac{1}{2} + \frac{1}{10}$. Now, we assume that there does not exist an equal partition.

We first notice that $j^*$ must be assigned to agent $1$ in the optimal solution. Otherwise, the principal's revenue is at most $\sum_{j \in \mathcal S} r_j p_{1,j} + r_{j^*} p_{2, j^*} + r_{\bar{j}} p_{1, \bar{j}} = \frac{1}{5} + \frac{1}{10} + \frac{1}{10} =  \frac{2}{5}$.

Moreover, since agent $2$ has probability $0$ on task $\bar{j}$, the optimal contract assigns $\bar{j}$ to agent $1$, which will not cause envy from agent $2$ but will improve the principal's revenue.

We have the following observation for the optimal solution.
\begin{observation}\label{EF_githt_ob}
    In an optimal solution, the EF constraints are tight for both agents.
\end{observation}
\begin{proof}
    Recall from the above arguments that  tasks $j^*$ and $\bar{j}$ are assigned to agent $1$ in the optimal solution.

    Suppose the optimal solution assigns $G_1 = \mathcal{S}_1 \cup\{j^*, \bar{j}\}$ and $G_2 = S_2$ to agent $1$ and $2$ respectively, where $\mathcal{S}_1 \cup \mathcal{S}_2 = \mathcal{S}$. Notice that the two agents have cost $0$ for all the tasks $j\in \mathcal{S}$. This implies that in the EF constraints of a feasible solution, agent $2$ gets nonnegative utility by exerting all the tasks in $G_1$ (without worrying about the costs which are {\it zero}), and similarly for agent $1$. Therefore, if either one of the EF constraints is not binding in the optimal solution (e.g., the EF constraint for agent $1$), one can further decrease some contract of some task in $G_1$, which will increase the principal's revenue, contradicting optimality. Note that such a decreasing step is always possible, as the minimum possible utility for the agent $1$ from $G_1$ is {\it zero}. Similar arguments also apply to agent $2$. This concludes the proof.
\end{proof}

Hence, by Observation \ref{EF_githt_ob}, we have the binding EF constraints as
\[
\sum_{j\in \mathcal{S}_1} r_j \alpha_j p_{1,j} + (r_{j^*} \alpha_{j^*}p_{1, j^*} - c_{1,j^*}) + r_{\bar{j}} \alpha_{\bar{j}} p_{1, \bar{j}} =  \sum_{j\in \mathcal{S}_2} r_j \alpha_j p_{1,j} \]
\[
\sum_{j\in \mathcal{S}_2} r_j \alpha_j p_{2,j} =\sum_{j\in \mathcal{S}_1} r_j \alpha_j p_{2,j} + r_{j^*} \alpha_{j^*}p_{2, j^*}
\]
We show that in the optimal contracts, $\alpha_j = 0$ for all $j\in \mathcal{S}_1$. Suppose this is not true. Then, given the optimal solution, we can construct a new solution with $\alpha_j'=0$ for all $j\in \mathcal{S}_1$ and for all $j\in \mathcal{S}_2$, let $\alpha_j' = \alpha_j - \Delta_j$ for some $\Delta_j \ge 0$ such that the EF constraints of agent $2$ is still binding,
\begin{equation}\label{33333}
    \sum_{j\in \mathcal{S}_2} r_j \alpha_j' p_{2,j} = r_{j^*} \alpha_{j^*}p_{2, j^*}
\end{equation}
which implies that $\sum_{j\in \mathcal{S}_2} r_j\Delta_jp_{2,j} = \sum_{j\in \mathcal{S}_1} r_j\alpha_j p_{2,j}$. Due to $p_{1,j} = 2p_{2,j}$ for all $j \in S$, the EF for agent $1$ is still binding,
\begin{equation}\label{44444}
 (r_{j^*} \alpha_{j^*}p_{1, j^*} - c_{1,j^*}) + r_{\bar{j}} \alpha_{\bar{j}} p_{1, \bar{j}} =  \sum_{j\in \mathcal{S}_2} r_j \alpha_j' p_{1,j}
 \end{equation}
This will lead to better principal's expected revenue, which is a contradiction. Hence, in the optimal contracts, $\alpha_j = 0$ for all $j\in \mathcal{S}_1$.

To maximize the principal's revenue, since all the utilities in $\mathcal{S}_1$ will be extracted by the principal ($\alpha_j = 0$ for $j\in \mathcal{S}_1$), it will be optimal to maximize the principal's revenue from tasks $j^*$ and  $\bar{j}$, and minimize the payments in $\mathcal{S}_2$ to agents $2$.

Recall that $\alpha_{j^*}\ge \frac{1}{2}$ by the effort constraint of agent $1$. By (\ref{33333}) and (\ref{44444}), we have the relation for $\alpha_{j^*}$ and $\alpha_{\bar{j}}$ in the optimal solution as
\[
\frac{4}{5}\alpha_{j^*} + \alpha_{\bar{j}}\frac{1}{10}=\frac{1}{2}
\]
Hence, to maximize the principal's revenue from tasks $j^*$ and $\bar{j}$, i.e.,
\[(1-\alpha_{j^*})r_{j^*}p_{1, j^*} + (1-\alpha_{\bar{j}})r_{\bar{j}}p_{1, \bar{j}}
\]
we can set $\alpha_{j^*}=\frac{1}{2}$ and $\alpha_{\bar{j}}=1$, which indeed also minimizes the payments to agent $2$ according to (\ref{33333}) and $\alpha_{j^*} \ge \frac{1}{2}$. Hence, the optimal solution would have $\alpha_{j^*}=\frac{1}{2}$ and $\alpha_{\bar{j}}=1$, which implies that by (\ref{33333}),
\[
\sum_{j\in \mathcal{S}_2} r_j \alpha_j p_{2,j} = \frac{1}{20}
\]
By our previous arguments,  to maximize the principal's expected revenue, we finally need to maximize the revenue
\[
\sum_{j\in \mathcal{S}_1} (1-\alpha_j) r_j p_{1, j} + \sum_{j \in \mathcal{S}_2} (1-\alpha_j) r_j p_{2,j} = \sum_{j\in \mathcal{S}_1}  r_j p_{1, j} + \sum_{j \in \mathcal{S}_2}  r_j p_{2,j} - \frac{1}{20}
\]
which is equivalent to maximizing $\sum_{j\in \mathcal{S}_1}  r_j p_{1, j} + \sum_{j \in \mathcal{S}_2}  r_j p_{2,j}$ while ensuring $ \sum_{j \in \mathcal{S}_2}  r_j p_{2,j} - \frac{1}{20}\ge 0$ due to that the principal revenue from agent $2$ is nonnegative.

Note that if $
\sum_{j\in \mathcal{S}_2} r_jp_{2,j} =\frac{1}{20} $, it
implies
$
\sum_{j\in \mathcal{S}_2} n_j = \frac{1}{2}T
$, i.e., an equal partition. Hence,  we must have
\[
\sum_{j\in \mathcal{S}_2} r_jp_{2,j} >\frac{1}{20}
\]
which due to $p_{1, j} > p_{2, j}$, implies that
\[
\sum_{j\in \mathcal{S}_1} (1-\alpha_j) r_j p_{1, j} + \sum_{j \in \mathcal{S}_2} (1-\alpha_j) r_j p_{2,j} = \sum_{j\in \mathcal{S}_1}  r_j p_{1, j} + \sum_{j \in \mathcal{S}_2}  r_j p_{2,j} - \frac{1}{20} <\frac{1}{10}.
\]
Hence, the total revenue will be
\[
(1-\alpha_{j^*})r_{j^*}p_{1, j^*} + (1-\alpha_{\bar{j}})r_{\bar{j}}p_{1, \bar{j}} + \sum_{j\in \mathcal{S}_1} (1-\alpha_j) r_j p_{1, j} + \sum_{j \in \mathcal{S}_2} (1-\alpha_j) r_j p_{2,j} <
 \frac{1}{2} + \frac{1}{10}.
\]

\end{proof}

\subsection{Proof of \cref{hardness_ef1_epsef}}

The proof is adapted from the proof of \cref{hardness_ef}. We prove the hardness results for EF1 contracts and $\eps$-EF contracts separately, which are summarized in \cref{thm:ef1_constant_hard} and \cref{app_them_eps_hard} below.

\begin{proposition}\label{thm:ef1_constant_hard}
When there are three agents, assuming {\rm \P~$\neq$~\NP}, for every constant $\rho>0$ there does not exist a polynomial-time algorithm that computes an EF1 contract with revenue at least a $\left(\frac{7}{10}+\rho\right)$ fraction of $\optefw$.
\end{proposition}

\begin{proof}
 We prove the results by reducing from partition: Given a set $S$ of integers $\{n_j\}$, find a subset of elements such that their sum equals $\frac{\sum_{j\in S} n_j}{2}$.

    \paragraph{Construction}
    There are three agents, i.e.,  $n=3$. Let $C=10$.
    Define two tasks $j^*$ and $\hat{j}$ with reward $r_{j^*} = r_{\hat{j}}=1$. The success probabilities are  $p_{1,j^*}=p_{1,\hat{j}}=1$ and $p_{2,j^*}=p_{3,j^*}=p_{2,\hat{j}}=p_{3,\hat{j}}=\frac{1}{C}$, and the costs are
    $c_{1,j^*} = c_{1,\hat{j}}=\frac{1}{2}$ and $c_{2,j^*}=c_{3,j^*}=c_{2,\hat{j}}=c_{3,\hat{j}}=0$.

    Let $M=C \sum_{j\in S}n_j $.
    For each $n_j$, we define a task $j$ with reward $r_j = \frac{n_j}{M}$. The success probabilities are $p_{1,j}=0$, $p_{2,j}=p_{3,j}=1$. The costs are $c_{1,j}=0$, $c_{2,j}=c_{3,j}=0$.

    \paragraph{Sufficiency.}
    Suppose there is a set of indexes $S^*$ such that $\sum_{j\in S^*} n_j=\frac{\sum_{j\in S} n_j}{2}$.
    Assign both tasks $j^*, \hat{j}$ to agent $1$ with $\alpha_{j^*}=\alpha_{\hat{j}}=1/2$.
    Moreover, assign all the tasks in $S^*$ to agent $2$, and  all the other tasks in $S\setminus S^*$ to agent $3$.
    Set all the contracts $\alpha_j=1$ to all tasks $j\in S$.

    Next, we verify that the allocation is EF1. Indeed, agent $1$ gets $0$ utility from the allocated tasks, and if switching to the sets of tasks of agent $2$ or $3$, it gets utility $0$.

    Also, agent $2$ does not envy agent $3$, and it gets the same utility if switching to the set of tasks assigned to agent $3$.
    By simple calculation, we have that agent $2$ gets utility $\sum_{j\in S^*} \frac{n_j}{M}=\frac{1}{20}$. By removing the task $j^*$ or $\hat{j}$ from the set assigned to agent $1$, agent $2$ gets utility $\frac{1}{20}$ from the remaining task. Hence, the EF1 constraints are satisfied. A similar argument also holds for agent $3$.

    Hence, the principal's revenue (only obtained from tasks $j^*, \hat{j}$) is $1$.

    \paragraph{Necessity.}
    We show that if there exists a contract with utility strictly greater than $\frac{1}{5} + \frac{1}{2}$,  then there exists a partition.

    First, notice that if at most one of the tasks $j^*, \hat{j}$ is assigned to agent $1$, the principal's revenue is at most:
    \[r_{j^*} p_{1, j^*} - c_{1,j^*}+ r_{\hat{j}}p_{2,\hat{j}} + \sum_{j\in S} r_j p_{2,j} = \frac{1}{2}+ \frac{1}{10}+ \sum_{j\in S} \frac{n_j}{M} =\frac{1}{5}+ \frac{1}{2}.\]
    Hence, in order to achieve optimal utility,  both tasks $j^*$ and $\hat{j}$ are assigned to agent $1$ and the contracts are set $\alpha_{j^*} , \alpha_{\hat{j}} \ge 1/2$.

    Since agent $2$ envies agent $1$ up to at most one task (EF1), it must be the case that $S_2$ (the set of tasks assigned to agent $2$) guarantees that
    \[\sum_{j \in S_2} \alpha_j r_j p_{2,j} = \sum_{j \in S_2} \alpha_j \frac{n_j}{M}\ge \alpha_{j^*} r_{j^*} p_{2,j^*} \ge \frac{1}{20}.\]
    which further implies
    \[\sum_{j \in S_2} \frac{n_j}{10 \sum_{j \in S}n_j}= \sum_{j \in S_2} \frac{n_j}{M}\ge \frac{1}{20}\]
    and
     \[\sum_{j \in S_2} n_j\ge \frac{1}{2} \sum_{j \in S} n_j.\]
     Similar results hold for agent $3$,
      \[\sum_{j \in S_3} n_j\ge \frac{1}{2} \sum_{j \in S} n_j.\]
The above together implies the existence of a partition. Thus any guarantee with revenue fraction strictly larger than $\frac{7}{10}$ would distinguish the two cases; equivalently, no approximation ratio strictly smaller than $\frac{10}{7}$ is possible.
\end{proof}

\begin{proposition}\label{app_them_eps_hard}
When there are three agents, assuming {\rm \P~$\neq$~\NP}, for every constant $\rho>0$ there does not exist a polynomial-time algorithm that computes an $\eps$-EF contract with revenue at least a $\left(\frac{7}{10}+\rho\right)$ fraction of $\opteps$, for any $\eps < \frac{1}{20}$.
\end{proposition}

\begin{proof}
 We prove the results by reducing from partition: Given a set of integers $\{n_1, \dots, n_{\ell}\}$, find a subset of elements such that their sum equals $\frac{\sum_{j \in [ \ell ]} n_j}{2}$.

    \paragraph{Construction}
    There are three agents, i.e.,  $n=3$.
    Define two tasks $j^*$ and $\hat{j}$ with reward $r_{j^*} = r_{\hat{j}}=1$. The success probabilities are  $p_{1,j^*}=p_{1,\hat{j}}=1$ and $p_{2,j^*}=p_{3,j^*} = \frac{1}{10}, p_{2,\hat{j}}=p_{3,\hat{j}}=2\eps$, and the costs are
    $c_{1,j^*} = c_{1,\hat{j}}=\frac{1}{2}$ and $c_{2,j^*}=c_{3,j^*}=c_{2,\hat{j}}=c_{3,\hat{j}}=0$.

    Let $M=10 \sum_{j\in [\ell]}n_j $.
    For each $n_j$, we define a task $j$ with reward $r_j = \frac{n_j}{M}$. The success probabilities are $p_{1,j}=0$, $p_{2,j}=p_{3,j}=1$. The costs are $c_{1,j}=0$, $c_{2,j}=c_{3,j}=0$.

    \paragraph{Sufficiency.}
    We show that if there exists a partition, then the principal's revenue is at least $1$.

Suppose that there is a set of indexes $\mathcal{S}$ such that $\sum_{j\in \mathcal{S}} n_j=\sum_{j \in [\ell]} n_j/2$.
Consider the following contract. Assign task $j^*$ and $\hat{j}$ to agent $1$ with contracts as $\alpha_{j^*}=\alpha_{\hat j}=\frac{1}{2}$.
Moreover, assign all the tasks in $\mathcal{S}$ to agent $2$, and  all the other tasks to agent $3$.
We set the contracts
$\alpha_j=1$ to all the tasks $j\in [\ell]$.

We are left to show that the allocation satisfies the effort constraints and guarantees revenue at least $1$.
First, the allocation is $\eps$-EF. Indeed, by the construction above, agent $1$ receives utility $0$. Due to $p_{1,j}=0$ for all $j\in [\ell]$, agent $1$ receives utility $0$ switching to the tasks assigned to one of the other agents.
Agent $2$ receives utility
\[\sum_{j\in \mathcal{S}} r_j p_{2,j} - c_{2,j} =  \sum_{j\in \mathcal{S}} \frac{n_j}{M}=\frac{1}{20}.\]
Switching to the tasks assigned to agent $1$, agent $2$ gets utility \[\alpha_{j^*}r_{j^*}p_{2, j^*}  + \alpha_{\hat{j}}r_{\hat{j}}p_{2, \hat{j}} = \frac{1}{20} + \eps, \]
which implies it is $\eps$-EF.
In addition, switching to the task assigned to agent $3$, agent $2$ gets utility $$\sum_{j \in [\ell]\setminus \mathcal{S}} \alpha_{j}r_{j}p_{2, j} - c_{2, j} = \frac{1}{20}$$ because $\mathcal{S}$ is a partition.
Hence, the agent $2$ does not envy other agents.
A similar argument applies to agent $3$.

We conclude noticing that the principal gets positive revenue only from task $j^*, \hat{j}$ and in particular $1$.

    \paragraph{Necessity.}
    We show that if there exists a contract with utility strictly greater than $\frac{7}{10}$,  then there exists a partition.

    First, notice that if at most one of the tasks $j^*, \hat{j}$ is assigned to agent $1$, the principal's revenue is at most:
    \[r_{\hat j} p_{1, \hat{j}} - c_{1,\hat{j}}+ r_{j^*}p_{2,j^*} + \sum_{j\in [\ell]} r_j p_{2,j} = \frac{1}{2}+ \frac{1}{10}+ \sum_{j\in [\ell]} \frac{n_j}{M} =\frac{7}{10}.\]
    Hence, in order to achieve optimal utility,  both tasks $j^*$ and $\hat{j}$ are assigned to agent $1$ and the contracts are set $\alpha_{j^*} , \alpha_{\hat{j}} \ge 1/2$.

    Since agent $2$ envies agent $1$ up to $\eps$, it must be the case that $S_2$ (the set of tasks assigned to agent $2$) guarantees that
    \[\sum_{j \in S_2} \alpha_j r_j p_{2,j} = \sum_{j \in S_2} \alpha_j \frac{n_j}{M}\ge \alpha_{j^*} r_{j^*} p_{2,j^*} + \alpha_{\hat{j}} r_{\hat{j}} p_{2,\hat{j}} -\eps \ge \frac{1}{20} .\]
    which further implies
    \[\sum_{j \in S_2} \frac{n_j}{10 \sum_{j \in \{1, 2, \dots, \ell\}}n_j}= \sum_{j \in S_2} \frac{n_j}{M}\ge \frac{1}{20}\]
    and
     \[\sum_{j \in S_2} n_j\ge \frac{1}{2} \sum_{j \in [\ell]} n_j.\]
     Similar results hold for agent $3$,
      \[\sum_{j \in S_3} n_j\ge \frac{1}{2} \sum_{j \in [\ell]} n_j.\]
The above together implies the existence of a partition. Thus any guarantee with revenue fraction strictly larger than $\frac{7}{10}$ would distinguish the two cases; equivalently, no approximation ratio strictly smaller than $\frac{10}{7}$ is possible.
\end{proof}

\section{Missing Proofs and Results in \cref{section_resutionced_alg}}
\label{apx:proofs}

\subsection{Proof of \cref{thm:constant_agent}}\label{app_proof_theorm_additve_epsef}

\begin{algorithm}[t]
\caption{Computing $3\xi$-EF Contracts \\
\textbf{Input:} $\xi>0$ ,  $\{r_j\}_{j \in \mathcal{M}}$, $\{p_{i,j}\}_{i \in \mathcal{N}, j \in \mathcal{M}}$, $\{c_{i,j}\}_{i \in \mathcal{N}, j \in \mathcal{M}}$, $m, n$ \\
\textbf{Output:} The contracts $(\tilde{\alpha}, \tilde{S})$.
}\label{algorithm_dynmaic_pr}
$\delta \gets \frac{\xi}{m}$ \;
Define  $\mathcal{D}_j$, $\mathcal{U}_i$, $\mathcal{H}$ as per Equation \eqref{eq:discretization_epsilon_EF} and functions  $\hat U_{i,j}(\cdot)$ and $\hat U^P(\cdot,\cdot)$ (Eq. \eqref{eq:disc_utility_EF})\;
$\mathcal{C}\gets$ Run Algorithm \ref{algorithm_dynmaic_pr2} with discretization  $\mathcal{D}_j$, $\mathcal{U}_i$, $\mathcal{H}$ and functions  $\hat U_{i,j}(\cdot)$, $\hat U^P(\cdot,\cdot)$\;
$\mathcal{C}^* \gets \emptyset$ \;
\For{ $( S,\alpha) \in \mathcal{C}$}{ \label{line_5_eps}
   \If{$(S, \alpha)$ satisfies $3\xi$-EF }{
$\mathcal{C}^* = \mathcal{C}^* \cup \{{S}, \alpha\}$ \;
   }
} \label{line9_eps}
\Return $({\tilde\alpha}, {\tilde S}) \in \mathcal{C}^*$ with maximum  principal's expected revenue.
\end{algorithm}

\begin{proof}
We prove a subroutine guarantee with internal parameter $\xi>0$: Algorithm \ref{algorithm_dynmaic_pr} returns a $3\xi$-EF contract with revenue at least $\optval-2\xi$. Invoking the subroutine with $\xi=\epsilon/3$ gives an $\epsilon$-EF contract with revenue at least $\optval-2\epsilon/3\ge \optval-\epsilon$, proving the theorem.
Let denote an optimal solution to Program (\ref{pr:fair}) with $\alpha^*= (\alpha_1^*, \alpha_2^*, \dots, \alpha_m^*)$, and its corresponding fair allocations of tasks $S^*= (S_1^*, S_2^*, \dots, S_n^*)$.

The proof proceeds as follows. We build a discretized instance of the problem, we solve it with dynamic programming, and we show that one of the solutions returned by Algorithm \ref{algorithm_dynmaic_pr2} is approximately EF and approximately optimal.

\paragraph{Discretized instance.}
We start building the discretized instance.
Given an internal parameter $\xi>0$, we define $\delta=\frac{\xi}{m}$. We discretize the set of contracts, agents' utilities, and the principal's expected revenue with a grid of step $\delta$. Formally, we define
\begin{equation}\label{eq:discretization_epsilon_EF}
    \begin{aligned}
        &\mathcal{D}_j= \{0,\delta, \ldots, 1\} \quad \forall j \in \mathcal{M}\\
        &\mathcal{U}_{i}=\{0,\delta, \ldots, 1\} \quad \forall i \in \mathcal{N}\\
        &\mathcal{H}=\{0,\delta, \ldots, 1\}\\
    \end{aligned}
\end{equation}

In the following since all sets $\mathcal{D}_j$ and $\mathcal{U}_i$ are the same for $i\in \mathcal{N}$ and $j\in \mathcal{M}$, we will refer to these sets simply with $\mathcal{D}$ and $\mathcal{U}$, respectively.
Finally, we define
\begin{equation}\label{eq:disc_utility_EF}
\begin{aligned}
&\hat U_{i,j}(\hat \alpha_j)= \lceil \hat \alpha_j p_{i,j} r_j-c_{i,j}\rceil_{\mathcal{U}} \quad \forall  i \in \mathcal{N}, j \in \mathcal{M}, \hat \alpha_{j}\in \mathcal{D}\\
&\hat U^P(\hat \alpha_j,i)=\lceil (1-\hat \alpha_j) p_{i,j} r_j\rceil_{\mathcal{H}}\quad  \forall  i \in \mathcal{N}, j \in \mathcal{M}, \hat \alpha_{j}\in \mathcal{D}
\end{aligned}
\end{equation}

\paragraph{Existence of a good solution with discretized contracts}

Consider the contract $\bar \alpha$ obtained setting  $\bar \alpha_i= \lceil \alpha^*_i\rceil_{\mathcal{D}}$. We show that $(S^*,\bar \alpha)$ is almost optimal and almost EF.
Formally:
\begin{equation}\label{discreteize_withinsec42_eps}
\begin{split}
    \sum_{k \in S^*_i}\bar{\alpha}_k p_{i, k }r_k - c_{i,k} &\ge \sum_{k \in S^*_i}\alpha_k^* p_{i, k }r_k - c_{i,k}\\
    & \ge \sum_{k \in S^*_j} \max\{ \alpha_k^* p_{i, k }r_k - c_{i,k}, 0\}\\
    & \ge \sum_{k \in S^*_j} \max\{ (\bar{\alpha}_k -\delta) p_{i, k }r_k - c_{i,k}, 0\} \\
    & \ge \sum_{k \in S^*_j} \max\{ \bar{\alpha}_k p_{i, k }r_k - c_{i,k}, 0\} -\xi\quad \quad \forall i,j \in \mathcal{N}
\end{split}
\end{equation}
where the second inequality is due to $(\alpha^*, S^*)$ satisfying the EF constraints. Moreover, the effort constraints for the tasks continue to hold as we round up the contracts.

Finally, we get that the principal's revenue under the new (approximate) contracts $(\bar{\alpha}, S^*)$ is at least $\optval-\xi$, which is
\begin{equation}\label{opt_approximation}
    \sum_{i \in \mathcal{N}} \sum_{k \in S^*_i} (1-\bar{\alpha}_k) p_{i,k}r_k \ge \sum_{i \in \mathcal{N}} \sum_{k \in S^*_i} (1-{\alpha}_k^* -\delta) p_{i,k}r_k \ge \sum_{i \in \mathcal{N}} \sum_{k \in S^*_i} (1-{\alpha}_k^* ) p_{i,k}r_k -\xi = \optval -\xi.
\end{equation}

\paragraph{Find a good solution with discretized contracts.}
The pseudocode of our algorithm is in Algorithm \ref{algorithm_dynmaic_pr}. Our algorithm enumerates over all the possible  pairs $(S,\alpha) \in \mathcal{C}$ returned by Algorithm \ref{algorithm_dynmaic_pr2}, and outputs the $3\xi$-EF contract that maximizes the principal's expected revenue.
Denote this contract by $(\tilde S, \tilde \alpha)$.

By Proposition \ref{thm:DP}, Algorithm \ref{algorithm_dynmaic_pr2} finds a set $\mathcal{C}$ which includes an allocation $\hat S$ and a contract $\hat \alpha$ that induce the same discretized utility profile as $(\bar \alpha,S^*)$.
Formally, we have that
 \begin{align*}
        &\sum_{i \in \mathcal{N}} \sum_{j \in S_i^*} \hat U^P( \bar \alpha_j,i)= \sum_{i \in \mathcal{N}} \sum_{j \in \hat S_i} \hat U^P(\hat \alpha_j,i)\\
        &\sum_{k \in  S_j^*} \hat U_{i,k}(\bar\alpha_k) = \sum_{k \in \hat S_j} \hat U_{i,k}(\hat \alpha_k) \forall i,j\in \mathcal{N}\\
        & \hat \alpha_j p_{i,j} r_j-c_{i,j} \ge 0 \quad \forall i \in \mathcal{N}, j \in \hat S_i.
\end{align*}

The last step of the proof shows that $(\tilde S, \tilde \alpha)$ achieves principal's expected revenue at least $\optval -2\xi$. Indeed, this contract is $3\xi$-EF by construction, while all the contracts in $\mathcal{C}$ satisfy the effort constraints.
To do so, it is sufficient to show that $(\hat S, \hat \alpha)$ is $3\xi$-EF and $2\xi$-optimal, since this contract belongs to $\mathcal{C}$.

We start by showing that the contract is approximately EF. Take two agents $i$ and $j$.
Then
\begin{align*}
    \sum_{k \in \hat S_i} \hat \alpha_k p_{i,k} r_k - c_{i,k} & \ge \sum_{k \in \hat S_i} (\hat U_{i,k}(\hat \alpha_k) -\delta) \\
    &\ge \sum_{k \in \hat S_i} \hat U_{i,k}(\hat \alpha_k) -\xi \\
    &= \sum_{k \in  S^*_i}  \hat U_{i,k}(\bar \alpha_k) -\xi\\
    & \ge \sum_{k \in S^*_i}\bar{\alpha}_k p_{i, k }r_k - c_{i,k} -\xi\\
    & \ge \sum_{k \in S^*_j} \max\{ \bar{\alpha}_k p_{i, k }r_k - c_{i,k}, 0\} -2\xi\\
    & \ge \sum_{k \in S^*_j} \hat U_{i,k}(\bar \alpha) -3\xi\\
    &= \sum_{k \in \hat S_j} \hat U_{i,k}(\hat \alpha) -3\xi\\
    & \ge \sum_{k \in \hat S_j} \max\{ 0, \hat \alpha_k p_{i,k} r_k - c_{i,k} \}- 3\xi.
\end{align*}
where the second and fourth inequalities are by the definition of $\hat{U}_{i,k}$, the fifth inequality is by (\ref{discreteize_withinsec42_eps}) and the equalities are by Proposition \ref{thm:DP}.
Moreover, it is $2\xi$-optimal since
\begin{align*}
\sum_{i \in \mathcal{N}} \sum_{k \in \hat S_i} (1-\hat{\alpha}_k) p_{i,k}r_k &\ge \sum_{i \in \mathcal{N}} \sum_{k \in \hat S_i} U^P(\hat\alpha_{k},i)- \xi\\
& = \sum_{i \in \mathcal{N}} \sum_{k \in S^*_i} U^P(\bar \alpha_{k},i) -\xi\\
& \ge \sum_{i \in \mathcal{N}} \sum_{k \in S^*_i} (1-\bar{\alpha}_k) p_{i,k}r_k - \xi\\
& \ge \optval-2\xi.
\end{align*}
where the last inequality is by
(\ref{opt_approximation}).

We conclude the proof by noting that the algorithm runs in polynomial time. Indeed, Algorithm \ref{algorithm_dynmaic_pr} clearly runs in polynomial time excluding the calls to Algorithm \ref{algorithm_dynmaic_pr2}. Moreover, it is easy to see that each call to Algorithm \ref{algorithm_dynmaic_pr2} runs in polynomial time since the possible vectors of agents' cumulative utilities and the principal's cumulative expected revenue have polynomial size (i.e., each utility belongs to $\{0,\delta, \ldots,m\}$).
\end{proof}

\subsection{Proof of \cref{lemma4_6fptasef1}}

\begin{proof}

First, we prove that
\begin{align}\label{eq:initialEF1}
\max\{\bar{\alpha}_j p_{i, j }r_j - c_{i,j},0\}\le \max\{{\alpha}^*_j p_{i, j }r_j - c_{i,j},0\} + \delta U_i \quad  \forall i \in \mathcal{N}, j \in \mathcal{M}.
\end{align}

Consider an agent $i$ and a task $j$. We consider three cases:
\begin{enumerate}[(i)]
    \item  $\underline{\alpha}_j\ge \tau_{i,j}$ and ${\alpha}^*_j<\tau_{i,j}$
    \item $\underline{\alpha}_j\ge \tau_{i,j}$ and ${\alpha}^*_j \ge \tau_{i,j}$
    \item $\underline{\alpha}_j< \tau_{i,j}$.
\end{enumerate}

(i) If $\underline{\alpha}_j\ge \tau_{i,j}$ and ${\alpha}^*_j< \tau_{i,j}$, then both  $\max\{\bar{\alpha}_j p_{i, j }r_j - c_{i,j},0\} $ and $\max\{{\alpha}^*_j p_{i, j }r_j - c_{i,j},0\}$ are $0$, as $\bar \alpha_j = \lceil \alpha^*_j\rceil_{\mathcal{D}_j} \le \tau_{i,j}$ by that the set $\mathcal{D}_{i,j}$ starts from $\tau_{i,j}$.

(ii) If $\underline{\alpha}_j\ge c_i/(r_j p_{i,j})$ and ${\alpha}^*_j\ge c_i/(r_j p_{i,j})$, then
\begin{align*}
\max\{\bar{\alpha}_j p_{i, j }r_j - c_{i,j},0\}&= \bar{\alpha}_j p_{i, j }r_j - c_{i,j} \\
& \le  ({\alpha}^*_j+ \delta (\underline{\alpha}_{j}  - \tau_{i,j})) p_{i, j }r_j -c_{i,j}  \\
& \le  {\alpha}^*_j p_{i, j }r_j+ \delta (\underline{\alpha}_{j} p_{i, j }r_j  - c_{i,j}) -c_{i,j}  \\
& \le  {\alpha}^*_j p_{i, j }r_j-c_{i,j} + \delta U_i  \\
&=\max\{{\alpha}^*_j p_{i, j }r_j - c_{i,j},0\} + \delta U_i,
\end{align*}
where the last inequality follows from the definition of $\underline{\alpha}_{i,j}$ in Equation \eqref{def:barAlpha} and $\underline{\alpha}_j \le \underline{\alpha}_{i,j}$.

(iii) If $\underline{\alpha}_j< c_i/(r_j p_{i,j})$, then $\max\{\bar{\alpha}_j p_{i, j }r_j - c_{i,j},0\} $ and $\max\{{\alpha}^*_j p_{i, j }r_j - c_{i,j},0\}$ are $0$, as $\bar \alpha_j = \lceil \alpha^*_j\rceil_{\mathcal{D}_j} \le \underline{\alpha}_j$ by Lemma \ref{wellefeindinelemmandk} and the definition of $\underline{\alpha}_j$. This concludes the proof of Equation \eqref{eq:initialEF1}.

Then, we use Equation \eqref{eq:initialEF1} to prove Equation \eqref{discreteize_withinsec42}. Consider two agents $i$ and $j$. It holds:
\begin{align*}
    \sum_{k \in S^*_i}\bar{\alpha}_k p_{i, k }r_k - c_{i,k} &\ge \sum_{k \in S^*_i}\alpha_k^* p_{i, k }r_k - c_{i,k}\\
    & \ge \sum_{k \in S^*_j} \max\{ \alpha_k^* p_{i, k }r_k - c_{i,k}, 0\}\\
    & \ge \sum_{k \in S^*_j} (\max \{\bar\alpha_{k} p_{i,k}r_k -c_{i,k},0\}-\delta U_i)\\
     & \ge \sum_{k \in S^*_j} \max \{\bar\alpha_{k} p_{i,k}r_k -c_{i,k},0\}-\nu U_i \quad\forall i,j \in \mathcal{N},
\end{align*}
where the third inequality is by Equation \eqref{eq:initialEF1}.
Moreover, the effort constraints for the tasks continue to hold as we round up the contracts.

Finally, we get that the principal's revenue under the new approximate contracts $(\bar{\alpha}, S^*)$ is at least $\optval-\nu$, which is
\begin{align*}
\sum_{i \in \mathcal{N}} \sum_{k \in S^*_i} (1-\bar{\alpha}_k) p_{i,k}r_k \ge \sum_{i \in \mathcal{N}} \sum_{k \in S^*_i} (1-{\alpha}_k^* -\delta) p_{i,k}r_k \ge \sum_{i \in \mathcal{N}} \sum_{k \in S^*_i} (1-{\alpha}_k^* ) p_{i,k}r_k -\nu = \optval -\nu.
\end{align*}
\end{proof}

\subsection{Proof of \cref{lemma4_7_epftsef1}}

\begin{proof}
It holds
\begin{align*}
\sum_{i \in \mathcal{N}} \sum_{k \in \hat S_i} (1-\hat{\alpha}_k) p_{i,k}r_k &\ge \sum_{i \in \mathcal{N}} \sum_{k \in \hat S_i} ( \hat{U}^P(\hat\alpha_{k},i)- \delta)\\
&\ge \sum_{i \in \mathcal{N}} \sum_{k \in \hat S_i}  \hat{U}^P(\hat\alpha_{k},i)- \nu\\
& = \sum_{i \in \mathcal{N}} \sum_{k \in S^*_i} \hat{U}^P(\bar \alpha_{k},i) -\nu\\
&\ge \sum_{i \in \mathcal{N}} \sum_{k \in S^*_i} (1-\bar{\alpha}_k) p_{i,k}r_k -\nu \\
& \ge \optval-2\nu.
\end{align*}
where the last inequality follows from Equation \eqref{eq:apx_opt_EF1}.
\end{proof}

\subsection{Proof of \cref{wellefeindinelemmandk}}
\begin{proof}
Consider a $j\in \mathcal{M}$. We show that $\alpha^*_j\le \underline{\alpha}_{j}$. This is sufficient to show that $\bar \alpha_j$ is  well-defined by the definition of $\mathcal{D}_j$.

We consider two cases: $\underline{\alpha}_j=1$ and $\underline \alpha_j<1$.

If $\underline{\alpha}_{j}=1$, it immediately implies $\alpha^*_j\le \underline{\alpha}_{j}$. Hence, in the following arguments, we focus on the second case $\underline{\alpha}_{j}<1$.

Suppose by contradiction that $\alpha^*_{j} > \underline{\alpha}_{j} $. Let $i$ be such that $\underline{\alpha}_j=\underline{\alpha}_{i,j}$. Then, by definition,
\begin{align}\label{eq:useDefAlpha}
    \underline{\alpha}_j= \frac{c_{i,j}+U_i}{q_{i,j}},
\end{align}
where we use $ \underline{\alpha}_{j}<1$.
Furthermore, suppose that
 $i'$ is the agent such that $j \in S^*_{i'}$.
Then, it holds:
    \begin{align*}
    U_i   &\ge \sum_{k \in S^*_{i}}\alpha^*_k p_{i, k }r_k - c_{i,k}\\
    & \ge \sum_{k \in S^*_{i'}}\max \{\alpha^*_k p_{i, k }r_k - c_{i,k},0\}\\
    &\ge  \alpha_j^* p_{i, j }r_j - c_{i,j} \\
    &>  \underline{\alpha}_{j} p_{i, j }r_j - c_{i,j}\\
    &= \frac{U_{i}+c_{i,j}}{p_{i,j}r_j}  p_{i, j }r_j -c_{i,j}\\
    & \ge U_{i},
    \end{align*}
where the first inequality is by Equation \eqref{eq:guess}, the second inequality is by EF and the equality follows from Equation \eqref{eq:useDefAlpha}. Hence, we reach a contradiction.
\end{proof}

\section{Missing Proofs and Results in \cref{sec:priceoffairness}}
\label{apx:additional}

\subsection{\texorpdfstring{Proof of \cref{prop:pof_eps_ef}}{Proof of epsilon-EF Price of Fairness Proposition}}
\label{apx:eps_ef_pof}

We first record a single-task guarantee that will be used in both upper-bound constructions. For a fixed task, suppress the task index and write $q_i=p_i r\le 1$. We use the break-even-share notation from the model: $\tau_i$ denotes the fixed-task version of $\tau_{i,j}$, so $\tau_i=c_i/q_i$ when $q_i>0$, $\tau_i=0$ when $q_i=c_i=0$, and $\tau_i=+\infty$ when $q_i=0<c_i$. For agents with $\tau_i\le 1$, let $W_i=q_i-c_i=q_i(1-\tau_i)$ be the principal's revenue from assigning the task to agent $i$ at her break-even share. Agents with $\tau_i>1$ never obtain positive utility from any linear share $\alpha\in[0,1]$, so they can be ignored in the search for a profitable assignee. Let $W=\max_{\tau_i\le 1} W_i$ be the unconstrained optimal revenue from the task.

\begin{lemma}\label{lem:single_task_eps_ef}
For any single task and any $0<\delta\le 1/4$, there exists a $\delta$-EF linear contract for that task with expected revenue at least $4\delta W$.
\end{lemma}

\begin{proof}
    If $W=0$, assign the task to any effort-feasible agent using share $\tau_i$. This gives revenue $0$ and is $\delta$-EF. Hence, assume $W>0$. 

    Start from an agent $a$ maximizing $W_i$. We sort the agents according to their break-even shares so that $\tau_1 \ge\tau_2 \ge \dots\ge \tau_n$. We can simply ignore those agents $i$ with $\tau_i >1$ as they will never be assigned the task by effort constraints. To ease exposition, we assume $\tau_1 \le 1$. 
    
    Start from the agent $a$ and set $\alpha=\tau_a$. 
Note that any agent $i<a$ gains non-positive utility at  $\alpha=\tau_a$, implying the satisfaction of effort constraints and $\delta$-EF constraints. Hence, the contract is $\delta$-EF exactly when every agent $i>a$ values it by at most $\delta$; for agents with $q_i>0$ with $i>a$, this condition is
    \[
    0\ge \max\{0, q_i \tau_a -c_i\} -\delta
    \]
    which, by $q_i \tau_i -c_i =0$, implies that
\[
q_i(\tau_a-\tau_i)_+\le \delta
\quad\text{for all such }i.
\] 
If there exists some agent $i>a$ such that the above condition fails, move to considering this agent $i$ and check the feasible conditions if it is assigned the task and has contract $\alpha = \tau_i$. Finally, the process terminates at some agent $b$ for which the break-even-share contract is $\delta$-EF. 

If the process terminates at $b = a$, then the revenue is $W$.

If $b \neq a$, then $b>a$. It must be that agent $b$ violated the above $\delta$-EF condition for all its immediate predecessor agents along the path. Let the path be $a \to \cdots\to i \to j \to b$. By construction, we know that agent $b$ violates the $\delta$-EF constraints when the contract is $\alpha = \tau_j$, i.e., 
\[
0 < \max\{q_b \tau_j - c_b, 0\} -\delta = q_b \tau_j - c_b - \delta,
\]
Moreover, since $\tau_i\ge \tau_j$, we also have that 
\[
0 < q_b \tau_i - c_b - \delta.
\]
Hence, when the contract is  $\alpha = \tau_a$ and the task is assigned to agent $a$, we have
\[
0 < q_b\tau_a -c_b -\delta
\]
which, by $q_b\tau_b-c_b =0$, implies that
\[
q_b(\tau_a-\tau_b)>\delta.
\]
Therefore,
\[
\frac{W}{W_b}
=
\frac{q_a(1-\tau_a)}{q_b(1-\tau_b)}
<
\frac{q_a(1-\tau_a)(\tau_a-\tau_b)}{\delta(1-\tau_b)}.
\]
Let $x=(\tau_a-\tau_b)/(1-\tau_b)$. Since $\tau_b<\tau_a\le 1$, we have $x\in[0,1]$, and the last expression equals
\[
\frac{q_a x(1-x)(1-\tau_b)}{\delta}
\le
\frac{1}{4\delta},
\]
using $q_a\le 1$ and $x(1-x)\le 1/4$. Thus $W_b\ge 4\delta W$, as required.

\end{proof}

We now prove the positive guarantee in \cref{prop:pof_eps_ef}. For each task $j$, let
\[
W_j=\max_{i\in\mathcal{N}}\{p_{i,j}r_j-c_{i,j}\}
\]
be its unconstrained contribution. Since tasks are independent without fairness constraints,
\[
\opt=\sum_{j\in\mathcal{M}} W_j.
\]

First, set $\delta=\eps/m$. Applying \cref{lem:single_task_eps_ef} independently to every task gives, for each task $j$, a single-task $\delta$-EF contract with revenue at least $4\delta W_j$. Combining these contracts yields a full allocation. For any ordered pair of agents $(i,\ell)$, each task in $\ell$'s bundle contributes at most $\delta$ to agent $i$'s envy, while agent $i$'s own assigned tasks give nonnegative utility. Hence
\[
U_i(S_\ell)-U_i(S_i)
\le |S_\ell|\delta
\le m\delta
=\eps.
\]
Thus the combined contract is $\eps$-EF and has revenue at least
\[
\sum_{j\in\mathcal{M}}4\delta W_j
=
\frac{4\eps}{m}\opt.
\]
This proves $\opteps\ge (4\eps/m)\opt$.

We next prove the bound depending on $n$. Apply \cref{lem:single_task_eps_ef} to every task with tolerance $\delta=\eps$. For each task $j$, let $a(j)$ be the assignee returned by the lemma, let $\alpha_j$ be the corresponding share, and let $R_j$ be the resulting principal revenue. Then every agent values task $j$ under share $\alpha_j$ by at most $\eps$, the designated assignee $a(j)$ receives zero utility from the task, and
\[
R_j\ge 4\eps W_j.
\]
Choose an agent $h$ whose designated tasks have total designated revenue at least the average:
\[
\sum_{j:a(j)=h}R_j
\ge
\frac{1}{n}\sum_{j\in\mathcal{M}}R_j
\ge
\frac{4\eps}{n}\opt.
\]

For tasks designated to $h$, keep the shares $\alpha_j$ above. For every other task $j$, reset its share to $\min\{\tau_{i,j}:\tau_{i,j}\le 1, \forall i\}$, the smallest break-even share among feasible agents for that task. Under these reset contracts, every agent obtains zero utility from every task not designated to $h$. Under the kept contracts, every task gives every agent utility at most $\eps$.

Now allocate tasks by the modified round-robin procedure used in the proof of \cref{ef1upperboundpro}, starting with agent $h$: on her turn, an agent chooses a remaining task of maximum utility among the tasks for which her effort constraint is satisfied, breaking ties in favor of the largest principal revenue; if no remaining task satisfies her effort constraint, she skips. Since every task is feasible for at least one agent under its chosen share, this procedure terminates with a full allocation. The standard round-robin argument gives EF1. Moreover, every single task has value at most $\eps$ to every agent, so EF1 implies $\eps$-EF.

It remains to lower-bound the revenue. Agent $h$ has zero utility for every task under the shares above, and all tasks designated to $h$ are feasible for her. Therefore, whenever $h$ takes a turn while one of her designated tasks remains, the tie-breaking rule gives her a task whose principal revenue is at least the revenue of any remaining task designated to $h$. Between two consecutive turns of $h$, at most $n-1$ other tasks are allocated. Thus each task taken by $h$ can be charged to at most $n$ tasks designated to $h$, each with no larger designated revenue. Consequently, the revenue obtained from tasks assigned to $h$ is at least
\[
\frac{1}{n}\sum_{j:a(j)=h}R_j
\ge
\frac{4\eps}{n^2}\opt.
\]
All other assigned tasks contribute nonnegative revenue, so
\[
\opteps
\ge
\frac{4\eps}{n^2}\opt.
\]
 Combining the two guarantees gives
\[
\opteps
\ge
\frac{4\eps}{\min\{m,n^2\}}\opt.
\]

Finally, we prove the lower-bound construction. Suppose $0<\eps<1/4$, and fix any $\rho>0$. Consider one nonzero task with reward $1$ and two agents $H$ and $L$. Let
\[
p_H=1,\qquad c_H=\frac12,
\]
and
\[
p_L=2\eps+\eta,\qquad c_L=0,
\]
where $\eta>0$ is chosen so that $2\eps+\eta<1/2$ and $2\eta\le \rho$. Without fairness, the principal assigns the task to $H$ at share $1/2$ and obtains revenue $1/2$. In any $\eps$-EF contract assigning the task to $H$, the effort constraint requires $\alpha\ge 1/2$, so agent $L$ values $H$'s contract by at least
\[
\alpha p_L
\ge
\frac{2\eps+\eta}{2}
=
\eps+\frac{\eta}{2}
>
\eps,
\]
while $L$'s own utility is zero. Hence such a contract is not $\eps$-EF. The task must therefore be assigned to $L$, whose revenue is at most $p_L=2\eps+\eta$. Thus
\[
\opteps
\le
2\eps+\eta
\le
(4\eps+\rho)\opt,
\]
where the last inequality uses $\opt=1/2$. For larger $n$ and $m$, add dummy agents and dummy zero-reward tasks; they do not affect either optimum. This completes the proof of \cref{prop:pof_eps_ef}.

\subsubsection{Another lower bound on the multiplicative gap}

The proposition gives an upper bound on the multiplicative gap $\opt/\opteps$. The construction below upper-bounds $\opteps$ for a family of instances, and therefore gives a lower bound on $\opt/\opteps$. Together, these bounds show that the dependence on $\eps$ is tight for fixed $n$ and $m$, while the exact dependence on $n$ and $m$ is non-trivial and remains open.

Fix an integer $1\le k\le n-1$ and assume $\eps\le m$. Let $\sigma=\sqrt{\eps/m}$. There are $m$ identical tasks, each with reward $1$, one high-productivity agent $H$, and $k$ low-productivity agents. If the instance has more than $k+1$ agents, the remaining agents are dummy agents with zero success probabilities and zero costs, and they do not affect either optimum. Set
\[
p_H=1,\qquad
c_H=t=\frac{1-\sigma}{1+k\sigma},
\]
and, for every low agent $\ell$,
\[
p_\ell=\sigma,\qquad c_\ell=0.
\]
The high agent's welfare per task is
\[
1-t=\frac{\sigma(k+1)}{1+k\sigma},
\]
which is at least the welfare $\sigma$ of any low agent. Hence
\[
\opt
=
m\frac{\sigma(k+1)}{1+k\sigma}.
\]

Consider any $\eps$-EF contract, and let $h$ be the number of tasks assigned to $H$. Every task assigned to $H$ must have share at least $t$, so every low agent values $H$'s bundle by at least $\sigma t h$. If $S_\ell$ is the bundle of low agent $\ell$, $\eps$-EF implies
\[
\sigma\sum_{j\in S_\ell}\alpha_j
\ge
\sigma t h-\eps.
\]
Therefore, the revenue from tasks assigned to low agent $\ell$ is at most
\[
\sum_{j\in S_\ell}\sigma(1-\alpha_j)
\le
\sigma |S_\ell|-\sigma t h+\eps.
\]
Summing over the $k$ low agents and adding the revenue from $H$'s tasks gives
\[
\rev
\le
\sigma m+k\eps
+h\bigl((1-t)-\sigma(1+kt)\bigr).
\]
By the choice of $t$, we have $1-t=\sigma(1+kt)$, and thus
\[
\opteps
\le
\sigma m+k\eps
=
m\sigma(1+k\sigma).
\]
Consequently,
\[
\frac{\opt}{\opteps}
\ge
\frac{k+1}{(1+k\sqrt{\eps/m})^2}.
\]
Optimizing over $k$ gives instances with
\[
\frac{\opt}{\opteps}
=
\Omega\!\left(\min\left\{n,\sqrt{\frac{m}{\eps}}\right\}\right).
\]

\end{document}